\documentclass[twocolumn,secnumarabic,amssymb, nobibnotes, aps, prx,nofootinbib]{revtex4-1}

\usepackage{amssymb}
\usepackage{subfigure}
\usepackage{graphicx}
\usepackage{amsmath}
\usepackage{color}
\usepackage{comment}
\usepackage{amsthm}
\usepackage{diagbox}
\usepackage{placeins}
\usepackage{hyperref}
\usepackage[capitalize]{cleveref}
\usepackage{braket}
\usepackage{newfloat,algcompatible}
\usepackage[size=small]{caption}
\usepackage{etoolbox}
\usepackage{ulem}

\AtEndEnvironment{algorithm}{\noindent\hrulefill\par\nobreak\vskip-8pt}

\DeclareFloatingEnvironment[
    fileext=loa,
    listname=List of Algorithms,
    name=ALGORITHM,
    placement=tbhp,
]{algorithm}
\DeclareCaptionFormat{algorithms}{\vskip-6pt\hrulefill\par#1#2#3\vskip-6pt\hrulefill}
\algblock[Input]{Input}{EndInput}
\algblockdefx[Input]{Input}{EndInput}
    [1]{\textbf{Input} #1}
    {}

\theoremstyle{definition}
\newtheorem{theorem}{Theorem}
\newtheorem{definition}[theorem]{Definition}

\newtheorem{lemma}[theorem]{Lemma}

\newcommand{\ie}{\textit{i.e.}}
\newcommand{\eg}{\textit{e.g.}}
\newcommand{\argmin}[1]{\mathop{\mathrm{argmin}}_{#1}}
\newcommand{\argmax}[1]{\mathop{\mathrm{argmax}}_{#1}}
\newcommand{\qspspace}[1]{\mathcal{C}_{#1}}
\newcommand{\qspobj}[2]{\Re\left[\left\langle0\left|U_{#1}\left({#2}\right)\right|0\right\rangle\right]}

\newcommand{\qspdeg}{d}

\renewcommand{\Re}{\mathrm{Re}}
\renewcommand{\Im}{\mathrm{Im}}
\newcommand{\I}{\mathrm{i}}
\newcommand{\polylog}{\mathrm{poly}\,\log}

\newcommand{\mc}[1]{\mathcal{#1}}

\newcommand{\wt}[1]{\widetilde{#1}}

\newcommand{\abs}[1]{\left\lvert#1\right\rvert}
\newcommand{\norm}[1]{\left\lVert#1\right\rVert}

\newcommand{\ud}{\,\mathrm{d}}
\newcommand{\Or}{\mathcal{O}}

\newcommand{\NN}{\mathbb{N}}
\newcommand{\RR}{\mathbb{R}}
\newcommand{\CC}{\mathbb{C}}

\newcommand{\seccheby}{R}

\newcommand{\BQIC}{Berkeley Center for Quantum Information and Computation, Berkeley, California 94720 USA}
\newcommand{\DeptMath}{Department of Mathematics, University of California, Berkeley, California 94720 USA}

\newcommand{\LBLMath}{Computational Research Division, Lawrence Berkeley National Laboratory, Berkeley, CA 94720, USA}
\newcommand{\DeptChem}{Department of Chemistry, University of California, Berkeley, California 94720 USA}

\usepackage[qm]{qcircuit}

\begin{document}
        
\title{Efficient phase-factor evaluation in quantum signal processing}

\author{Yulong Dong$^{1,2}$} 
\author{Xiang Meng$^{3}$}
\author{K. Birgitta Whaley$^{1,2}$}
\author{Lin Lin$^{3,4}$}

\affiliation{$^1$\BQIC}
\affiliation{$^2$\DeptChem}
\affiliation{$^3$\DeptMath}
\affiliation{$^4$\LBLMath}

\date{\today}

\begin{abstract}
  Quantum signal processing (QSP) is a powerful quantum algorithm to exactly implement matrix polynomials on quantum computers.  Asymptotic analysis of quantum algorithms based on QSP has shown that asymptotically optimal results can in principle be obtained for a range of tasks, such as Hamiltonian simulation and the quantum linear system problem.  A further benefit of QSP is that it uses a minimal number of ancilla qubits, which facilitates its implementation on near-to-intermediate term quantum architectures. However, there is so far no classically stable algorithm allowing computation of the phase factors that are needed to build QSP circuits. Existing methods require the usage of variable precision arithmetic and can only be applied to polynomials of relatively low degree. We present here an optimization based method that can accurately compute the phase factors using standard double precision arithmetic operations. We demonstrate the performance of this approach with applications to Hamiltonian simulation, eigenvalue filtering, and the quantum linear system problems. Our numerical results show that the optimization algorithm can find phase factors to accurately approximate polynomials of degree larger than $10,000$ with error below $10^{-12}$. 
\end{abstract}

\maketitle

\section{Introduction}
Recent progress in quantum algorithms has enabled construction of efficient quantum circuit representations for a large class of non-unitary matrices, which significantly expands the potential range of applications of quantum computers beyond the original goal of efficient simulation of unitary dynamics envisaged by Benioff \cite{Benioff1980} and Feynman \cite{Feynman1982}. The basic tool for representation of non-unitary matrices and hence of non-unitary quantum operators is called block-encoding \cite{GilyenSuLowEtAl2019}. It describes the process in which one embeds a non-unitary matrix $A$ into the upper-left block of a larger unitary matrix $U_A$, and then expresses the quantum circuit in terms of $U_A$. 

Computation of matrix functions, \ie,  evaluation of $F(A)$, where $F(x)$ is a smooth (real-valued or complex-valued) function,  is a central task in numerical linear algebra \cite{Higham2008}.  Numerous computational tasks can be performed by generating approximations to matrix functions.  These include application of a broad range of operators to quantum states: \eg, $e^{-\I t A}$ for the Hamiltonian simulation problem; $e^{-\beta A}$ for the thermal state preparation problem; $A^{-1}$ for the matrix inverse (also called the quantum linear system problem, QLSP); and the spectral projector of $A$ for the principal component analysis, to name a few.

Several routes to 
construct a quantum circuit for $f(A)$ have been developed.  These include methods using phase estimation (\eg, the HHL algorithm \cite{HarrowHassidimLloyd2009} for the matrix inverse), the method of linear combination of unitaries (LCU) \cite{ChildsKothariSomma2017,BerryChildsKothari2015}, and  the method of quantum signal processing (QSP)  \cite{LowYoderChuang2016,LowChuang2017,GilyenSuLowEtAl2019}. 
Among these methods, QSP stands out as so far the most general approach capable of representing a broad class of matrix functions via the eigenvalue or singular value transformations of $A$, while using a minimal number of ancilla qubits.
The basic idea of QSP is to approximate the desired function $F(x)$ by a polynomial function $f(x)$, and then find a circuit to encode $f(A)$ \textit{exactly} (assuming an exact block-encoding $U_A$). 
Treating the block-encoding $U_A$ as an oracle, the application of QSP has given rise to asymptotically optimal Hamiltonian simulation algorithms \cite{ChildsMaslovNamEtAl2018,HaahHastingsKothariEtAl2018}. 
Applications have also been made to solving QLSP  \cite{GilyenSuLowEtAl2019,Haah2019}, and to eigenvalue filtering \cite{LinTong2019}. In particular, the eigenvalue filtering approach of Ref. \cite{LinTong2019} does not directly approximate $A^{-1}$, but approximates a spectral projection operator, leading also to a quantum algorithm for solving QLSP with near-optimal complexity without the need of involving complex procedures such as variable time amplitude amplification \cite{Ambainis2012}.

Despite these fast growing successes, practical application of QSP on quantum computers, whether these are near- or long-term machines, still faces a significant challenge. A QSP circuit is defined using a series of adjustable phase factors. Once these phase factors are known, the QSP circuit can be directly implemented using $U_A$ together with a set of multi-qubit control gates and single qubit phase rotation gates. However, the inverse problem, \ie, finding the phase factors associated with a given polynomial function $f(x)$ is extremely difficult, to the extent that in practice very few applications of QSP have been made to date.
The original work of Low and Chuang \cite{LowChuang2017} demonstrated the existence of the phase factors but was not constructive. Initial efforts to find constructive procedures were not encouraging.  Thus it was reported in  \cite{ChildsMaslovNamEtAl2018} that it was prohibitive to obtain a QSP circuit of length that is larger than $30$ for the Jacobi-Anger expansion \cite{LowChuang2017} of the Hamiltonian simulation problem, and concluded ``the difficulty of computing the angles needed to perform the QSP algorithm prevents us from taking full advantage of the algorithm in practice, so it would be useful to develop a more efficient classical procedure for specifying these angles''. 

The first constructive procedure to find phase factors was given in \cite{GilyenSuLowEtAl2019}, with a procedure which requires a recursive solution of roots of high degree polynomials to high precision, counting multiplicities  of the roots. Therefore this procedure is not stable for representing high degree polynomials using QSP. 
Significant improvement has recently been made by Haah \cite{Haah2019},  who proposed a numerical algorithm to compute phase factors up to order $\sim 1000$, provided that all arithmetic operations can be computed with sufficiently high precision.  Specifically, the number of classical bits needed for this scales as $\Or(d \log (d/\epsilon))$, where $d$ is the degree of the polynomial $f$, and $\epsilon$ is the target accuracy. Therefore the algorithm is still not \textit{classically} numerically stable (a numerically stable algorithm should use no more than $\Or(\polylog (d/\epsilon))$ classical bits) \cite{Higham2002}. 
Haah's algorithm was implemented in Ref. \cite{Haah2019} using $\textsf{Mathematica}$ and employing the variable precision arithmetic capability of this.  The running time is observed to be  $\Or(d^3)$.

In this paper, we demonstrate that the phase factors can be accurately determined with standard double precision arithmetic operations, even when the degree of the polynomial $f(x)$ is very high ($\gtrsim 10,000$) and when a very high precision ($L^\infty$ error of function approximation $\lesssim 10^{-12}$) is required. We achieve this with a standard optimization approach that only minimizes a loss function, rather than recursively determining the phase terms.  This minimization involves the multiplication of matrices in SU(2) and is thus  numerically stable. We iteratively refine the phase factors to minimize the loss functions.  However, since the optimization of the phase factors is a very nonlinear procedure, the initial guess must be carefully chosen. 
Indeed, if we randomly select the initial guess, the accuracy of the resulting phase factors is usually very low.
We also find that under proper conditions, the QSP phase factors exhibit an inversion symmetry structure with respect to the center. This should be respected in the initial guess and preserved throughout the optimization procedure. We combine these two features to provide a simple, and yet highly effective choice of the initial guess. 

We demonstrate here the performance of our optimization based approach to determine the phases for QSP algorithms with examples for Hamiltonian simulation, eigenstate filtering, and matrix inversion. We show that our algorithm can significantly outperform existing approaches using variable precision arithmetic operations \cite{GilyenSuLowEtAl2018,Haah2019}. Numerical observation indicates that the computational cost of our method scales only quadratically as $\Or(d^2)$, while the number of classical bits used remains constant (using the standard double precision, \ie, 64 bits, arithmetic operations) as $d$ increases. 

We note that the previous algorithms for finding the phase factors require an analytic expansion of the smooth function $F(x)$ into polynomials. For instance, the Jacobi-Anger expansion is used for Hamiltonian simulation \cite{LowChuang2017,Haah2019}. When $F(x)$ is defined only on a sub-interval of $[-1,1]$, as for, \eg,  matrix inversion, where $F(x)=1/x$ is not well defined at $x=0$, one must first find an approximate smooth function and then perform expansion with respect to this approximate smooth function. Both steps introduce additional approximations and lead to inefficiencies in implementation. As an alternative, we propose here to use the Remez exchange algorithm \cite{Remez1934} to directly find the minimax approximation to $F(x)$ on $[-1,1]$ or a given sub-interval. Our numerical evidence shows that this not only streamlines the process of finding QSP factors, but that the use of the Remez algorithm can also lead to polynomials of significantly lower degree. 

Besides the inversion symmetry, we also find that the phase factors used for approximating smooth functions can decay rapidly away from the center. We find that the decay of the phase factors is directly linked to the decay of the coefficients in the Chebyshev expansion of the target function. This enables us to design a ``phase padding'' procedure, which identifies an initial guess of the QSP phase factors for a high degree polynomial, given the corresponding phase factors for a relatively low degree polynomial.

Throughout this paper we shall use the following notation:  $N=2^{n},M=2^m$, and $[N]=\Set{0,1,\ldots,N-1}$, with $n$ the number of logical qubits (also called system qubits), and $m$ the number of qubits added to construct the unitary $U_A$. We shall refer to the latter as the ``ancilla qubits for block-encoding'', which is to be distinguished with additional ancilla qubits needed for quantum signal processing.
$T_{d}$ and $\seccheby_{d}$ are Chebyshev polynomials of degree $d$ of the first and second kind respectively.
For a matrix $A$, the transpose, Hermitian conjugate and complex conjugate are denoted by $A^{\top}$, $A^{\dag}$, $A^*$, respectively. 

\section{Review of quantum signal processing}
\subsection{Block-encoding and qubitization}\label{sec:blockencode}

Block-encoding is a general technique to encode a non-unitary matrix on a quantum computer. Let $A\in \CC^{N\times N}$ be an $n$-qubit Hermitian matrix. If we can find an $(m+n)$-qubit unitary matrix $U\in\CC^{MN\times MN}$ such that
\begin{equation}
U_A=\left(\begin{array}{cc}
{A} & {\cdot} \\
{\cdot} & {\cdot}
\end{array}\right)
\label{eqn:block_encode_exact_matrix}
\end{equation}
holds, \ie, $A$ is the upper-left matrix block of $U_A$, then we may get access to $A$ via the unitary matrix $U_A$. In particular,
\begin{equation}
A=\left(\langle 0^m | \otimes I_n\right) U_A \left( | 0^m \rangle \otimes I_n \right).
\label{eqn:block_encode_exact}
\end{equation}
In general, the representation \eqref{eqn:block_encode_exact} may not
exist, \eg, when the operator norm $\norm{A}_2$ is larger than $1$.  So
the definition of block-encoding should be relaxed as follows
\cite{LowChuang2017,GilyenSuLowEtAl2019}: if we can find $\alpha, \epsilon \in \mathbb{R}_+$, a state $|G\rangle \in \mathbb{C}^M$, and an $(m+n)$-qubit matrix $U_A$ such that
\begin{equation}
\Vert A - \alpha \left(\langle G | \otimes I_n\right) U_A \left( | G \rangle \otimes I_n \right) \Vert \leq \epsilon,
\label{eqn:block_encoding}
\end{equation}
then $U_A$ is called an $(\alpha, m, \epsilon)$-block-encoding of $A$. Here $\ket{G}$ is referred to as the signal state (for block-encoding). Then \cref{eqn:block_encode_exact} gives a $(1,m,0)$-block-encoding of $A$ with $\ket{G}=\ket{0^m}$. If $U_A$ is Hermitian, it is called a Hermitian block-encoding. In particular, all the eigenvalues of a Hermitian block-encoding $U_A$ are $\pm 1$. 
For simplicity of presentation, in the following we present the explicit construction of block-encoding and qubitization for Hermitian $U_A$. We shall then briefly discuss the generalization to non-Hermitian $U_A$ and refer the reader to \cref{app:qsp_general_block_encode} for full details of this.

As an example, assume that $A$ is written as the linear combination of Pauli operators \cite{BerryChildsKothari2015,ChildsKothariSomma2017} with real coefficients, as
\begin{equation}
A=\sum_{l\in [M]} c_l U_l, \quad c_{l}\ge 0.
\label{eqn:lc_pauli}
\end{equation}
Here $U_{l}$ is a multi-qubit Pauli operator, which is unitary and Hermitian. We assume the availability of two oracles. The first one is the $(m+n)$-qubit select oracle:
\begin{equation}
U_{\mathrm{SEL}}=\sum_{l\in [M]} \ket{l}\bra{l}\otimes U_{l}.
\label{eqn:prepare_oracle}
\end{equation}
$U_{\mathrm{SEL}}$ implements the selection of the unitary $U_{l}$ on conditioned on the state of the $m$-qubit signal register. 
The second is the $m$-qubit prepare oracle that generates a specific superposition of the $m$-qubit signal states (note that $\ket{l=0}\equiv \ket{0^m}$):
\begin{equation}
U_{\mathrm{PREP}}\ket{0^m}=\frac{1}{\sqrt{\norm{c}_1}}\sum_{l\in[M]} \sqrt{c_l}\ket{l},
\label{eqn:select_oracle}
\end{equation}
where the $1$-norm is $\norm{c}_1=\sum_{l\in[M]}\abs{c_l}$. 
Then defining
\begin{equation}
U_A=(U_{\mathrm{PREP}}^{\dagger}\otimes I_n) U_{\mathrm{SEL}}(U_{\mathrm{PREP}}\otimes I_n),
\label{eqn:UA_Pauli}
\end{equation}
we may verify that $U_A$ is a $(\norm{c}_1,m,0)$-Hermitian block encoding of $A$. 

We also define
\begin{equation}
U_{\Pi}=2\ket{0^m}\bra{0^m}\otimes I_n-I_m\otimes I_n.
\label{eqn:U_Pi}
\end{equation}
Both $U_{\Pi}$ and  $U_A$ are unitary and Hermitian. Then Jordan's lemma \cite{Jordan1875} states that the entire Hilbert space $\mc{H}=\CC^{MN}$ can be decomposed into orthogonal subspaces $\mc{H}_{j}$ invariant under $U_{\Pi}$ and $U_A$,
where each $\mc{H}_{j}$ has dimension 1 or 2. Restricted to each irreducible two-dimensional subspace $\mc{H}_j$, with a properly chosen basis denoted by $\mc{B}_j$, the matrix representations of $U_{\Pi}$ and $U_A$ are
\begin{equation}
[U_{\Pi}]_{\mc{B}_j}=\left(\begin{array}{cc}{1} & {0} \\ {0} & {-1}\end{array}\right), 
[U_A]_{\mc{B}_j}=\left(\begin{array}{cc}{\lambda_j} & {-\sqrt{1-\lambda_j^{2}}} \\ {- \sqrt{1-\lambda_j^{2}}} & {-\lambda_j}\end{array}\right).
\label{eqn:twoblock_representation}
\end{equation}
Here $\lambda\in[-1,1]$, and a potential phase factor in the off diagonal elements of $[U_A]_{\mc{B}_j}$ can be absorbed into the choice of the basis. It is worth noting that we can always choose  $[U_{\Pi}]_{\mc{B}_j}$ to be a $\sigma_z$ matrix. Given the eigendecomposition $\alpha^{-1}A=\sum_{j\in[N]} \lambda_j \ket{\psi_j}\bra{\psi_j}$, there are exactly $N$ such two-dimensional subspaces $\mc{H}_j$ of the full Hilbert space $\mc{H}$. Each subspace is associated with a vector $\ket{0^m}\ket{\psi_j}$ in the $(m+n)$-qubit space and \cref{eqn:twoblock_representation} gives
\begin{equation}
(\bra{0^m}\otimes I_n)U_A\ket{0^m}\ket{\psi_j}=\lambda_j\ket{\psi_j}=\alpha^{-1}A\ket{\psi_j}.
\label{eqn:eigen_UA}
\end{equation}
Each subspace $\mc{H}_j$ is also the invariant subspace of the operator $\Omega := U_{\Pi} U_A$, which is referred to as the \textit{iterate} \cite{low2019hamiltonian}. Furthermore, when restricted to $\mc{H}_j$, the iterate $\Omega$ is a rotation matrix with eigenvalues $e^{\pm \I \arccos(\lambda_j)}$. 
Then the combined space $\oplus_{j\in[N]}\mc{H}_j$ forms a $2N$-dimensional subspace of $\mc{H}$. This introduces an additional ancillary qubit, so that the total number of qubits is now $n+m+1$. Each eigenvalue $\lambda_j$ is associated with two branches and hence with an $\text{SU}(2)$ matrix via the mapping $\lambda_j=\cos \theta_j\mapsto e^{\pm \I \theta_j}$. This technique is called qubitization \cite{low2019hamiltonian}. 

Although the decomposition in \cref{eqn:twoblock_representation} formally involves the eigenvalue $\lambda_j$ of $A$ and the proper basis $\mc{B}_j$, it is important that we do not necessarily need the eigendecomposition of $A$ explicitly. In fact, the key advantage of qubitization is that one can perform the eigenvalue transformations for all eigenvalues simultaneously by means of the quantum signal processing approach.

\subsection{Quantum signal processing}\label{sec:qsp}

Given the above constructions of block-encoding and qubitization, quantum signal processing (QSP) then considers the following parameterized circuit consisting of $\qspdeg$ iterates and $\qspdeg+1$ rotations that are interleaved in alternating sequence:
\begin{equation}
U_{\wt{\Phi}}=\left[\prod_{i=0}^{\qspdeg-1}(e^{\I \tilde{\phi}_i U_{\Pi}} U_{\Pi} U_A)\right] e^{\I \tilde{\phi}_\qspdeg U_{\Pi}}.
\label{eqn:QSP_representation_1}
\end{equation}
Here $\tilde{\phi}_i\in\RR$, and $\wt{\Phi}=(\tilde{\phi}_0,\ldots,\tilde{\phi}_\qspdeg)$ is the vector of phase factors that will specify the polynomial $f(x)$ approximating the desired function $F(x)$. The use of the notation $\tilde{\phi}$ here is due to the fact that there are multiple sets of phase factors, which can be deduced from each other. In this section we use different notations such as $\tilde{\phi},\phi,\varphi$ to distinguish these phase factors, and record their relation explicitly.

We now summarize the construction of these phase factors for a non-unitary but Hermitian operator $A$, according to the approach of Ref. \cite{GilyenSuLowEtAl2019}.
For any $\tilde{\phi}\in\RR$ and $n$-qubit state $\ket{\psi}$, we have 
$$
e^{\I \tilde{\phi} U_{\Pi}}  \ket{0^m}\ket{\psi}=e^{\I \tilde{\phi}} \ket{0^m}\ket{\psi}.
$$
For any $m$-qubit state $\ket{\perp^m}$ satisfying $\braket{0^m|\perp^m}=0$, we have
$$
e^{\I \tilde{\phi} U_{\Pi}} \ket{\perp^m}\ket{\psi}=e^{-\I \tilde{\phi}}\ket{\perp^m}\ket{\psi}.
$$
Therefore
\begin{displaymath}
U_{\Pi}=-\I e^{\I \frac{\pi}{2} U_{\Pi}}.
\end{displaymath}
We may then absorb $U_{\Pi}$ into the rotation matrix as
\begin{equation}
U_{\wt{\Phi}}=(-\I)^\qspdeg\left[\prod_{i=0}^{\qspdeg-1}(e^{\I \varphi_i U_{\Pi}}  U_A) \right]e^{\I \varphi_\qspdeg U_{\Pi}}.
\label{eqn:QSP_representation_2}
\end{equation}
Here we have redefined the phase factors as $\varphi_i=\tilde{\phi}_i+\frac{\pi}{2}$ for $i=0,\ldots,\qspdeg-1$, and $\varphi_\qspdeg=\tilde{\phi}_\qspdeg$.
The global phase factor $(-\I)^\qspdeg$ can be optionally discarded and we shall do so below.

Then we may readily check that the matrix $e^{\I \varphi U_{\Pi}} $ has a $(1,1,0)$-block-encoding as illustrated in \cref{fig:phase_circuit}.
Here the control gate represents an $(m+1)$-qubit Toffoli gate (with the usual convention that open circles represent the target qubit being flipped when the control bits are zero).  

\begin{figure}[htb]
\begin{center}
\[
\Qcircuit @C=0.7em @R=1.0em {
\lstick{\ket{0}} & \targ & \gate{e^{-\I \varphi \sigma_z}} & \targ & \qw\\
& \ctrlo{-1} & \qw & \ctrlo{-1} & \qw\\
{\inputgroupv{2}{4}{.8em}{1.5em}{\ket{0^m}}}& \ctrlo{-1} & \qw & \ctrlo{-1} & \qw\\
& \ctrlo{-1} & \qw & \ctrlo{-1} & \qw\\
&\qw&\qw&\qw&\qw\\
{\inputgroupv{5}{7}{.8em}{1.em}{\ket{\psi}}}&\qw&\qw&\qw&\qw\\
&\qw&\qw&\qw&\qw\\
}
\]
\end{center}
\caption{Quantum circuit for block-encoding $e^{\I \varphi U_{\Pi}}$. The three distinct groups of lines represent $1$,$m$,$n$ qubits, respectively. }
\label{fig:phase_circuit}
\end{figure}
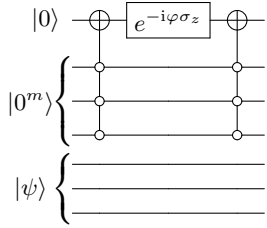

Using the circuit in \cref{fig:phase_circuit}, we may then implement the $(n+m)$-qubit unitary operator $U_{\wt{\Phi}}$ of  \cref{eqn:QSP_representation_2} using only one additional ancilla qubit and the circuit in \cref{fig:qsp_circuit}  \cite{GilyenSuLowEtAl2018}.

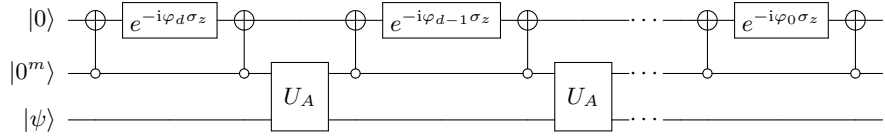
\begin{figure*}[htb]
\begin{center}
\[
\Qcircuit @C=0.7em @R=1.0em {
 \lstick{\ket{0}}& \targ & \gate{e^{-\I \varphi_{\qspdeg} \sigma_z}} & \targ & \qw & \targ & \gate{e^{-\I \varphi_{\qspdeg-1} \sigma_z}} & \targ & \qw & \qw &\raisebox{0em}{$\cdots$}&&\qw  &\targ & \gate{e^{-\I \varphi_0 \sigma_z}} & \targ & \qw  \\
\lstick{\ket{0^m}}& \ctrlo{-1} & \qw & \ctrlo{-1} & \multigate{1}{U_A} & \ctrlo{-1} & \qw & \ctrlo{-1} & \multigate{1}{U_A} &\qw &\raisebox{0em}{$\cdots$} &&\qw & \ctrlo{-1} & \qw & \ctrlo{-1} & \qw \\
\lstick{\ket{\psi}}&\qw&\qw&\qw& \ghost{U_A} &\qw&\qw&\qw& \ghost{U_A}&\qw &\raisebox{0em}{$\cdots$} &&\qw&\qw&\qw& \qw&\qw\\
}
\]
\end{center}
\caption{Quantum circuit for quantum signal processing of a general matrix polynomial with a Hermitian block-encoding $U_A$. }
\label{fig:qsp_circuit}
\end{figure*}

Ref. \cite{GilyenSuLowEtAl2018} investigated the general question as to which class of functions can be block-encoded by $U_{\wt{\Phi}}$ for some choice of phase factors. First, each $\mc{H}_j$ is an invariant subspace of $U_{\wt{\Phi}}$. So the upper-left element of $U_{\wt{\Phi}}$ acting on $\mc{H}_j$ is a function of the eigenvalue $\lambda_j$. Thus we see that qubitization reduces the problem of  representing a matrix function on an $n$-qubit system to a representation problem in $\mathrm{SU}(2)$, which can be carried out on classical computers.  We now state main theorem of QSP from Ref. \cite{GilyenSuLowEtAl2018} below in \cref{thm:qsp}. 
\begin{theorem}{(\textbf{Quantum Signal Processing in SU(2)} \cite[Theorem 3]{GilyenSuLowEtAl2018})}\label{thm:qsp}
    For any $P, Q \in \mathbb{C}[x]$ and a positive integer $\qspdeg$ such that \textit{(1)} $\deg(P) \leq \qspdeg, \deg(Q) \leq \qspdeg-1$, \textit{(2)} $P$ has parity $(\qspdeg\mod2)$ and $Q$ has parity $(\qspdeg-1 \mod 2)$, \textit{(3)} $|P(x)|^2 + (1-x^2) |Q(x)|^2 = 1, \forall x \in [-1, 1]$. Then, there exists a set of phase factors $\Phi := (\phi_0, \cdots, \phi_\qspdeg) \in [-\pi, \pi)^{\qspdeg+1}$ such that
\begin{equation}
\label{eq:qsp-gslw}
\begin{aligned}
        &U_\Phi(x) = e^{\I \phi_0 \sigma_z} \prod_{j=1}^{\qspdeg} \left[ W(x) e^{\I \phi_j \sigma_z} \right]\\
        &= \left( \begin{array}{cc}
        P(x) & \I Q(x) \sqrt{1 - x^2}\\
        \I Q^*(x) \sqrt{1 - x^2} & P^*(x)
        \end{array} \right)
\end{aligned}
\end{equation}
where 
\begin{displaymath}
W(x) = e^{\I \arccos(x) \sigma_x}=\left(\begin{array}{cc}{x} & {\I \sqrt{1-x^{2}}} \\ {\I \sqrt{1-x^{2}}} & {x}\end{array}\right).
\end{displaymath}
\end{theorem}

The proof of \cref{thm:qsp}  is constructive and, as shown explicitly in  Ref. \cite{GilyenSuLowEtAl2018}, it yields an algorithm to compute the phase factor vector $\Phi$ once the polynomials $P,Q\in \CC[x]$ are given. The algorithm of Ref. \cite{GilyenSuLowEtAl2018} is summarized in \cref{app:implement} (\cref{alg:GSLW}, with modifications to enhance the numerical stability). We note that these phase factors are unique, modulo certain trivial equivalence relations (\cref{app:qsp_unique}).

In order to connect \cref{thm:qsp} with the representation of $U_{\tilde{\Phi}}$ in \cref{eqn:QSP_representation_1}, we consider the matrix representation of $U_{\wt{\Phi}}$ restricted to $\mc{H}_j$, let $x=\lambda_j$, and use the following identity
\begin{equation}
e^{\I \arccos(x) \sigma_x}=e^{-\I \frac{\pi}{4} \sigma_z} \left(\begin{array}{cc}{x} & {- \sqrt{1-x^{2}}} \\ { \sqrt{1-x^{2}}} & {x}\end{array}\right)e^{\I \frac{\pi}{4} \sigma_z}.
\label{eqn:Omega_W}
\end{equation}
Hence to connect \cref{eq:qsp-gslw} with \cref{eqn:QSP_representation_1}, we have $\tilde{\phi}_0=\phi_0-\frac{\pi}{4}$, $\tilde{\phi}_\qspdeg=\phi_\qspdeg+\frac{\pi}{4}$, and $\tilde{\phi}_i=\phi_i$ for $1\le i\le \qspdeg-1$. Therefore, the relation between the phase factors $\{\phi_i\}_{i=0}^\qspdeg$ in \cref{thm:qsp} (\cref{eq:qsp-gslw}) and the phase factors $\{\varphi_i\}_{i=0}^\qspdeg$ appearing in 
$U_{\tilde{\Phi}}$ of \cref{eqn:QSP_representation_1} and in the implementation of the QSP circuit in \cref{fig:qsp_circuit}, is given by
\begin{equation}
\varphi_i=\begin{cases}
\phi_0+\frac{\pi}{4}, & i=0,\\
\phi_i+\frac{\pi}{2}, & 1\le i\le \qspdeg-1,\\
\phi_n+\frac{\pi}{4}, & i=\qspdeg.\\
\end{cases}
\label{eqn:phi_varphi_relation}
\end{equation}

\subsection{Representing general matrix polynomials}\label{sec:matrixpolynomial}

Now given a degree $\qspdeg$ polynomial $P(x)\in\CC[x]$ satisfying the requirement of \cref{thm:qsp}, for any $(\alpha,m,0)$ Hermitian-block-encoding of $A$, the circuit in \cref{fig:qsp_circuit} 
yields a $(1,m+1,0)$-block-encoding of $P(A/\alpha)$. With some abuse of notation, we shall denote both this block-encoding of the polynomial function of $A$ and the associated QSP circuit by $U_{\Phi}$. The QSP circuit uses $\qspdeg$ queries of $U_A$ and $\Or((m+1)\qspdeg)$ other primitive quantum gates. 

We should remark that the condition (3) in \cref{thm:qsp} imposes very strong constraints on $P,Q$ that are nontrivial to satisfy. Therefore we consider the following cases separately on how to construct QSP circuits in practice. 

\textit{Case 1}. In many applications, we are interested in computing $f(A/\alpha)$, where $f(x)$ is a real polynomial. It is stated in \cite[Theorem 5]{GilyenSuLowEtAl2018} that for $f \in \mathbb{R}[x]$ satisfying \textit{(1), (2)} and $|f(x)| \leq 1, \forall x \in [-1, 1]$, there exists $P\in \mathbb{C}[x], Q \in \mathbb{R}[x]$ such that $\Re [P(x)] = f(x)$. The choice of $P,Q$ may not  be unique. This only gives the block-encoding of $P(A/\alpha)$. In order to obtain the block-encoding of $f(A/\alpha)$, we can use the linear combination of unitaries (LCU) technique to separate the real and imaginary parts of $P(x)$ as follows. Note that
\begin{equation}
f(x)=\frac12(P(x)+P^*(x)).
\end{equation}
If the upper-left entry of $U_{\Phi}(x)$ is $P(x)$ as in \cref{eq:qsp-gslw}, then
\[
\begin{aligned}
        &U^*_\Phi(x) = e^{-\I \phi_0 \sigma_z} \prod_{j=1}^{\qspdeg} \left[ W^*(x) e^{-\I \phi_j \sigma_z} \right]\\
        &= \left( \begin{array}{cc}
        P^*(x) & -\I Q^*(x) \sqrt{1 - x^2}\\
        -\I Q(x) \sqrt{1 - x^2} & P(x)
        \end{array} \right)
\end{aligned}
\]
Here $U_{\Phi}^*(x)$ is the complex conjugation of $U_{\Phi}(x)$, and hence its  upper-left entry of is $P^*(x)$.  From
\begin{displaymath}
W^*(x)=e^{\I \frac{\pi}{2} \sigma_z} W(x) e^{-\I \frac{\pi}{2} \sigma_z},
\end{displaymath}
we find that $U_{\Phi}^*(x)=U_{-\Phi}(x)$, where the negative phase factors are defined by
\begin{equation}
-\Phi:=\left(-\phi_0+\frac{\pi}{2},-\phi_1,\cdots,-\phi_{\qspdeg-1},-\phi_\qspdeg-\frac{\pi}{2}\right),
\label{eqn:negative_qsp}
\end{equation}
which simply negates each phase factor except for $\phi_0$ and $\phi_d$.
In order to find a block-encoding of $\frac12(U_{\Phi}+U_{-\Phi})$, we can introduce one additional ancilla qubit to the signal register. The prepare oracle $U_{\mathrm{PREP}}$ is simply the Hadamard gate $H$.  \cref{fig:lcu_real} gives the circuit for the $(1,m+2,0)$-block-encoding of $f(A/\alpha)$. This technique is also called the addition of block-encodings \cite{GilyenSuLowEtAl2019}. Note that according to \cref{eqn:phi_varphi_relation}, the negative phase factors $-\Phi$ should be implemented using  the  circuit
 in \cref{fig:qsp_circuit} with 
 \begin{equation}
\varphi_i=\begin{cases}
-\phi_0+\frac{3\pi}{4}, & i=0,\\
-\phi_i+\frac{\pi}{2}, & 1\le i\le \qspdeg-1,\\
-\phi_\qspdeg-\frac{\pi}{4}, & i=\qspdeg.\\
\end{cases}
\label{eqn:neg_phi_varphi_relation}
\end{equation}

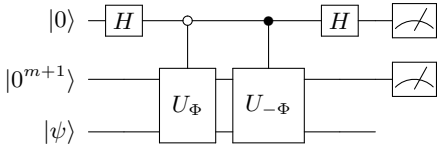
\begin{figure}[htb]
\begin{center}
\[
\Qcircuit @C=0.7em @R=1.0em {
\lstick{\ket{0}} & \gate{H} & \ctrlo{1} & \ctrl{1} & \gate{H} & \qw & \meter \\
\lstick{\ket{0^{m+1}}}&  \qw & \multigate{1}{U_{\Phi}} &
\multigate{1}{U_{-\Phi}}   & \qw&\qw & \meter \\
\lstick{\ket{\psi}}& \qw & \ghost{U_{\Phi}} &\ghost{U_{-\Phi}}&\qw&\qw \\
}
\]
\end{center}
\caption{Quantum circuit for block-encoding of $f(A/\alpha)$ using LCU to separate real and imaginary parts of $f(x)$. The three horizontal lines represent $1$,$m+1$,$n$ qubits, respectively. The circuits $U_{\Phi}$, $U_{-\Phi}$, are shown, after proper transformation of the phase factors, in \cref{fig:qsp_circuit}. }
\label{fig:lcu_real}
\end{figure}

\textit{Case 2}. The real polynomial $f(x)$ in case 1 is assumed to have definite parity. For a general real polynomial without parity constraints, we may use the decomposition
\begin{equation}
f(x)=f_{\mathrm{even}}(x)+f_{\mathrm{odd}}(x),
\label{eqn:evenodd}
\end{equation}
where $f_{\mathrm{even}}(x)=\frac12(f(x)+f(-x)), f_{\mathrm{odd}}(x)=\frac12(f(x)-f(-x))$. If $|f(x)|\le 1$ on $[-1,1]$, then  $|f_{\mathrm{even}}(x)|,|f_{\mathrm{odd}}(x)|\le 1$ on $[-1,1]$, and $f_{\mathrm{even}}(x)$, $f_{\mathrm{odd}}(x)$ can be each constructed using the circuit in \cref{fig:lcu_real}. Introducing another ancilla qubit and using the same form of the LCU circuit in  \cref{fig:lcu_real} (the $U_{\Phi},U_{-\Phi}$ circuits should be replaced by the QSP circuits for even and odd parts, respectively), we find a $(2,m+3,0)$-block-encoding of $f(A/\alpha)$. Equivalently, we have a $(1,m+3,0)$-block-encoding of $\frac12 f(A/\alpha)$.  

\textit{Case 3}. The most general case is that $f(x)\in\CC[x]$ is a complex polynomial. Let $f(x)=g(x)+\I h(x)$ where $g,h\in \RR[x]$ are the real and imaginary parts of $f(x)$, respectively. We remark that even when $h=0$ (\ie, $f(x)$ is a real polynomial), the associated polynomial $P(x)$ might have a non-vanishing imaginary component. Therefore in general we cannot expect to find phase factors that simultaneously encode $g(x)+\I h(x)$, even if $f(x)$ has definite parity. Hence we need to use LCU once again to find the block-encoding of $f$ through the linear combination of block-encodings of $g$ and $\I h$, respectively. Assuming $|g(x)|,|h(x)|\le 1$ on $[-1,1]$, following case 2, we have a $(2,m+3,0)$-block-encoding of $g(A/\alpha)$ denoted by $U_g$. Similarly a circuit of the form in \cref{fig:lcu_imag} gives the $(2,m+3,0)$-block-encoding of $\I h(A/\alpha)$ denoted by $U_{\I h}$.

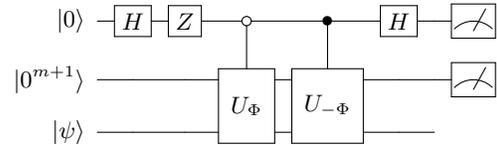
\begin{figure}[htb]
\begin{center}
\[
\Qcircuit @C=0.7em @R=1.0em {
\lstick{\ket{0}} & \gate{H} & \gate{Z} &\ctrlo{1} & \ctrl{1} & \gate{H} & \qw & \meter \\
\lstick{\ket{0^{m+1}}}&  \qw &  \qw& \multigate{1}{U_{\Phi}} &
\multigate{1}{U_{-\Phi}} &  \qw  & \qw & \meter \\
\lstick{\ket{\psi}}& \qw &  \qw& \ghost{U_{\Phi}} &\ghost{U_{-\Phi}}&\qw&\qw \\
}
\]
\end{center}
\caption{Quantum circuit for block-encoding of $\I h(A/\alpha)$ using linear combination of unitaries. The three lines represent $1$,$m+1$,$n$ qubits, respectively. The circuit $U_{\Phi}$, $U_{-\Phi}$, after proper transformation of the phase factors, is given in \cref{fig:qsp_circuit}. }
\label{fig:lcu_imag}
\end{figure}

We can use the LCU circuit of the form in \cref{fig:lcu_real}, with the $U_{\Phi},U_{-\Phi}$ circuits now replaced by $U_g$ and $U_{\I h}$, respectively, to ensure that the prepare oracle is still the Hadamard gate. This gives a $(4,m+4,0)$-block-encoding of $f(A/\alpha)$. 

We now make some general remarks on the block-encoding of matrix polynomials. First, while LCU is a general technique for implementing addition of block-encodings, when block-encoding a real polynomial as in case 1 above, one can actually save an ancilla qubit by taking advantage of the special structure of QSP circuits (see \cref{app:qsp_real_save}). A similar implementation exists for an imaginary polynomial, using a $Z$ gate as in \cref{fig:lcu_imag}. This reduces the number of additional ancilla qubits by $1$ for all cases discussed above, and the number of ancilla qubits then matches the results in \cite{GilyenSuLowEtAl2019}. Second, although the concept of qubitization and QSP were introduced here for Hermitian block-encodings in order to make use of Jordan's lemma, all the constructions shown above can be generalized to non-Hermitian block-encodings. 
One possible procedure to achieve this is described in \cref{app:qsp_general_block_encode}, which requires only use of one additional ancilla qubit. We note here that an alternative procedure is to use the quantum singular value transformation, which removes the need of this ancilla qubit and leads to a slightly simpler circuit, as well as allowing treatment of the case when $A$ is not a Hermitian matrix \cite{GilyenSuLowEtAl2019}. For simplicity all further discussion in this paper assumes that an $(\alpha,m,0)$-block-encoding $U_A$ is available. When the block-encoding itself is not error-free, \ie, $U_A$ is an $(\alpha,m,\epsilon)$-block-encoding of $A$, the cumulative error in the QSP circuit can also be analyzed. We refer readers to \cite{GilyenSuLowEtAl2018,GilyenSuLowEtAl2019} for more details. 

\subsection{Direct methods for finding phase factors}

According to \cref{sec:matrixpolynomial}, case 1 is the most important step, since cases 2 and 3 can simply be obtained from applying case 1 repeatedly and using the LCU technique. In fact, the proof of Theorem 5 in \cite{GilyenSuLowEtAl2018} also provides a constructive method for finding the phase factors, as follows. Given a properly normalized  real polynomial  with definite parity $f(x)$, one may first reconstruct complementing polynomials $B(x),C(x)\in\RR[x]$ to form $P=f+\I B, Q=C$ satisfying the requirement in \cref{thm:qsp}. This can be done by solving all the roots (including multiplicities) of the polynomial $1 - f(x)^2$ \cite[Lemma  6]{GilyenSuLowEtAl2018}. Then one can use a reduction method to find the phase factors. This procedure will be referred to as the GSLW method.  This procedure is exact if all floating point arithmetic operations can be performed with infinite precision, but is numerically unstable with standard double precision arithmetic operations. One disadvantage of the GSLW method is that it is based on the Taylor expansion of high order polynomials, which can be numerically highly unstable when the degree of polynomials becomes large.

To improve the numerical stability of the GSLW method, another algorithm was proposed in \cite{Haah2019}, which we will refer to as the Haah method. In the Haah method, the polynomials defined on $[-1, 1]$ are mapped to the unit circle via the transformation $x\mapsto e^{\pm \I \arccos(x)}$, and then extended to the complex plane. Such treatment is equivalent to a Chebyshev polynomial expansion, which improves the numerical stability over the GSLW method which uses the standard basis $\{1,x,x^2,\ldots\}$. Then, a similar reduction procedure is used to deduce the phase factors. However, one still needs to find the roots of a polynomial of high degree, and the number of classical bits required for this is $\Or(\qspdeg \log \qspdeg)$, where $\qspdeg$ is the degree of polynomial. 

In both the GSLW method and the Haah method, the phase factors are obtained from a single shot calculation. Therefore we refer to them as the direct methods for finding phase factors. This is in contrast to the optimization based method to be introduced below, which finds the phase factors via an iterative procedure.

The performance of the GSLW method has also been improved by a more recent work \cite{chao2020finding} after this paper was posted. The improved method of \cite{chao2020finding} is still based on direct factorization of polynomials. However, it is found that the numerical stability can be empirically improved using a method called ``capitalization'', which adds a small perturbation to the leading order term of the target polynomial. Together with another technique called ``halving'', the method of  \cite{chao2020finding} can find a sequence with more than $3000$ phase factors with double precision arithmetic operations. This result indicates that the sensitivity of the phase factors with respect to perturbation of the target polynomials is still not well understood.  Our optimization-based algorithm below presents a very different approach to determining the phase factors, which can achieve machine precision directly without perturbing the target polynomials and which is thus not limited by stability of such procedures. We show that with the optimization approach up to 10,000 phase factors can be determined with error less than $10^{-12}$.

\section{Optimization based method for finding phase factors}
Both the GSLW and the Haah methods are limited by the usage of root-finding and matrix reduction procedure, which result in the numerical instability when the degree of polynomials becomes large. Here we consider an alternative strategy to find the phase factors, by direct minimization with respect to  a certain distance function,
\begin{equation}\label{eq:opt-dis}
    L(\Phi) := \mathrm{dist}\left\{\qspobj{\Phi}{x},f(x)\right\}.
\end{equation}
In practice, the distance function will be characterized by the mean squared loss over discrete sample points. When $L(\Phi^*)$ is zero, we obtain the desired phase factors through the minimizer $\Phi^*$. This strategy bypasses the difficulty of constructing the complementing polynomials that relies on the high-precision root-finding procedure. Because the computation of the gradient and the Hessian matrix of the objective function only involve the matrix multiplications in SU(2), which is a numerically stable procedure, the optimization scheme is expected to significantly improve the robustness of the algorithm. This will be verified by our numerical tests. It also ensures an efficient optimization.

In the following discussion,  we use  $P, Q$ as the polynomials involved in the QSP unitary matrix in \cref{eq:qsp-gslw}. Let $\qspspace{\qspdeg+1} \subset [-\pi, \pi)^{\qspdeg+1}$ be the irreducible set of phase factors with $\qspdeg+1$ entries. The pair of polynomials $P(x), Q(x)\in \CC[x]$ satisfying conditions in \cref{thm:qsp} determines a unique set of phase factors $\Phi\in\qspspace{\qspdeg+1}$ (see \cref{app:qsp_unique}).

We again only consider a properly normalized  real polynomial  with definite parity $f(x)$ as in case 1 of \cref{sec:matrixpolynomial}. Because the form of $Q(x)$ is not of interest, we may restrict $Q(x)\in \RR[x]$.
 
\subsection{Symmetry property of the phase factors}\label{sec:symmetry}

Given a set of QSP factors $\Phi$, let the inverse phase factors be defined as
\begin{equation}
\Phi^- = (\phi_\qspdeg, \phi_{\qspdeg-1}, \cdots, \phi_0).
\label{eqn:inverse_qsp}
\end{equation}
The inverse phase factors should not be confused with the negative phase factors $-\Phi$ in \cref{eqn:negative_qsp}.

\cref{thm:inv} states that when we choose $Q(x)$ to be a real polynomial, the phase factors are symmetric under inversion. 
\begin{theorem}[\textbf{Inversion Symmetry}]
    \label{thm:inv}
     1) If $\Phi = \Phi^-$, then  $Q \in \mathbb{R}[x]$. 2) If  $Q \in \mathbb{R}[x]$, then we may choose $\Phi\in\mc{C}_{\qspdeg+1}$ such that  $\Phi=\Phi^{-}$.
\end{theorem}
\begin{proof}
        1): Obviously,
        \begin{equation}
        \begin{split}
                U_{\Phi^-}(x) &= e^{\I \phi_\qspdeg \sigma_z} \prod_{j = 1}^\qspdeg \left[ W(x) e^{\I \phi_{\qspdeg-j} \sigma_z} \right] = U_\Phi(x)^{\top}\\
                &= \left(\begin{array}{cc}{P(x)} & {\I Q^*(x) \sqrt{1-x^{2}}} \\ {\I Q(x) \sqrt{1-x^{2}}} & {P^{*}(x)}\end{array}\right).
        \end{split}
        \end{equation}
        Then, the statement that $\Phi$ is invariant under inversion implies that $Q(x) = Q^*(x) \in \mathbb{R}[x]$.\\
       2): If $Q \in \mathbb{R}[x]$, then $U_\Phi(x)  =U_\Phi(x)^{\top}= U_{\Phi^-}(x)$. Expand $P, Q$ in terms of Chebyshev polynomials, \ie, $P(x) = \sum_j p_j T_j(x),\ \sqrt{1-x^2} Q(x) = \sum_j q_j \sqrt{1-x^2} \seccheby_{j-1}(x)$. After a change of variable $x = \cos\theta$, $P, Q$ are transformed to Fourier series in terms of $\cos(j\theta)$ and $\sin(j\theta)$ respectively. The continuation $\theta \mapsto 2\pi - \theta$ extends the QSP unitary consisting of $P, Q$ to a $\mathrm{U}(1) \rightarrow \mathrm{SU}(2)$ function, after identifying $\theta$ with $e^{\I \theta}\in U(1)$.  Moreover, the parity constraint implies that this function only has non-zero coefficients $j = -\qspdeg, -\qspdeg+2, \cdots, \qspdeg-2, \qspdeg$ with respect to $e^{\I j \theta}$. \cref{app:qsp_unique} shows that the set of phase factors is unique, up to the equivalence relation for the irreducible set $\qspspace{\qspdeg+1}$. So $\Phi=\Phi^-$ up to equivalence relations. In particular, we may choose the phase factors such that  $\Phi=\Phi^-$. 
\end{proof}
As an example, let $P(x) = T_\qspdeg(x),\ Q(x) = \seccheby_{\qspdeg-1}(x)$, the corresponding QSP phase factors are $\Phi = (\underbrace{0, 0, \cdots, 0}_{\qspdeg+1})$. For $P(x) = \I T_\qspdeg(x),\ Q(x) = \seccheby_{\qspdeg-1}(x)$, the phase factors are $\Phi = (\frac{\pi}{4}, \underbrace{0, \cdots, 0}_{\qspdeg-1}, \frac{\pi}{4})$. In both cases, the polynomial $Q$ is real. Thus, it is evident that the phase factors satisfy the inversion symmetry in \cref{thm:inv}.

The symmetry property allows us to reduce the number of degrees of freedom by a factor of 2, and also motivates the symmetric construction of phase factors in the optimization procedure later. The appearance of two $\pi/4$ factors in the example above can be justified by \cref{lma:rot}, which shows that the action of these phase factors interchanges the real and imaginary parts of the polynomial $P$ up to a sign.

\begin{lemma}
    \label{lma:rot}
    Given a set of QSP phase factors $\Phi$, the following relations hold point-wise for $x \in [-1, 1]$,
    \begin{equation*}
    \begin{split}
        &\Re\left[ \langle 0 | U_\Phi(x) | 0 \rangle \right] = - \Im\left[ \langle 0 | e^{- \I \frac{\pi}{4} \sigma_z} U_\Phi(x) e^{- \I \frac{\pi}{4} \sigma_z} | 0 \rangle \right],\\
        &\Im\left[ \langle 0 | U_\Phi(x) | 0 \rangle \right] = \Re\left[ \langle 0 | e^{- \I \frac{\pi}{4} \sigma_z} U_\Phi(x) e^{- \I \frac{\pi}{4} \sigma_z} | 0 \rangle \right].
    \end{split}
    \end{equation*}
\end{lemma}
\begin{proof}
    Factorize the QSP unitary as $U_\Phi = a_0 I + a_1 \sigma_z + a_2 \sigma_x + a_3 \sigma_y$. The algebra of Pauli matrices implies that $e^{- \I \frac{\pi}{4} \sigma_z} U_\Phi e^{- \I \frac{\pi}{4} \sigma_z} = a_0 e^{- \I \frac{\pi}{2} \sigma_z} + a_1 e^{- \I \frac{\pi}{2} \sigma_z} \sigma_z + a_2 \sigma_x + a_3 \sigma_y = - \I a_0 I- \I a_1 \sigma_z + a_2 \sigma_x + a_3 \sigma_y$. Then, the conclusion follows.
\end{proof}

\subsection{Choice of objective function}\label{sec:opt-cho}

If the target smooth function $f(x)$ is not a polynomial,  we first  approximate $f(x)$ using a polynomial, and then feed the polynomial into the QSP solver. We would stress that this preprocessing step of polynomial approximation is necessary for the success of the optimization method. If we directly feed a non-polynomial function $f(x)$ into the objective function, then generally the equation $L(\Phi) = 0$ does not have a solution. Numerical evidence indicates that the landscape of the objective function is very complex and the optimization procedure can easily get stuck in one of the many local minima. On the other hand, for any polynomial satisfying conditions in \cref{thm:qsp}, there always exists a set of QSP factors $\Phi^*$ so that $L(\Phi^*) = 0$. Our numerical results indicate that starting from a proper initial guess, the optimization procedure can be very robust.

Since $Q(x)$ is not involved in the distance function, we may require $Q(x)\in \RR[x]$ and impose the inversion symmetry constraint (\cref{thm:inv}) on the phase factors. Under this constraint, the phase factors $\Phi=(\phi_0,\dots,\phi_\qspdeg)$ have $\lceil\frac{\qspdeg+1}{2}\rceil$ degrees of freedom for optimization. As a result, it is reasonable to choose the approximation as a polynomial $f$ of degree $\qspdeg$ with parity $(\qspdeg\mod2)$, which has the same number of adjustable coefficients. \cref{thm:qsp} and \cref{thm:inv} together guarantee the existence of symmetric phase factors $\Phi$ such that $\Re\left[\langle0|U_{ \Phi}(\cdot)|0\rangle\right]=f$. In this case, the optimization over $\Phi$ towards the minimum value of the distance function can be viewed as a polynomial interpolation taking the QSP parameterization. These features suggest that the mean squared loss in terms of $\tilde \qspdeg:=\lceil\frac{\qspdeg+1}{2}\rceil$ sample points   on $(0, 1]$ provides an accurate enough characterization of distance function. Therefore, we can write objective function for optimization as
\begin{equation} \label{eq:alg-prac}
    L(\hat{\Phi}) = \frac{1}{\tilde \qspdeg} \sum_{j=1}^{\tilde \qspdeg} \left|\Re\left[\langle0|U_{ \Phi}(x_j)|0\rangle\right] - f(x_j)\right|^2,
\end{equation}
where for $\hat{\Phi}=(\phi_0,\dots,\phi_{\tilde \qspdeg-1})\in[-\pi,\pi)^{\tilde \qspdeg}$,
\begin{equation}
    \Phi=\left\{\begin{array}{ll}
    (\phi_0,\cdots,\phi_{\tilde \qspdeg-1},\phi_{\tilde \qspdeg-1},\cdots,\phi_0) & {\qspdeg\text{ is odd}}, \\
     (\phi_0,\cdots,\phi_{\tilde \qspdeg-2},\phi_{\tilde \qspdeg-1},\phi_{\tilde \qspdeg-2},\cdots,\phi_0) & {\qspdeg\text{ is even}}.
    \end{array}\right.
\end{equation}
We choose $x_j=\cos\left(\frac{(2j-1)\pi}{4\tilde \qspdeg}\right),\,j=1,\dots,{\tilde \qspdeg}$ as the positive roots of the Chebyshev polynomial $T_{2\tilde \qspdeg}(x)$. \cref{thm:optmodel} shows that using the Chebyshev nodes, the accuracy of the polynomial approximation can be directly measured in terms of the objective function (the proof is given in \cref{app:proof_cheby}). 

\begin{theorem}\label{thm:optmodel}
Suppose we have the following expansions:
$$
f(x)=\sum_{j=0}^\qspdeg \alpha_jT_j(x),\,\,f_\Phi(x)=\sum_{j=0}^\qspdeg \beta_jT_j(x),
$$
where $f_\Phi(x)=\qspobj{\Phi}{x}$. If the discrete samples are chosen to be positive roots of $T_{2\lceil\frac{\qspdeg+1}{2}\rceil}(x)$ and $L(\hat{\Phi})\leq \epsilon$, then we have 
$$
\max_{j=1,\dots,\qspdeg}|\alpha_j-\beta_j|\leq 2\sqrt{\epsilon}.
$$
\end{theorem}

Note that the optimal phase factors are not necessarily unique. This is because the real part of $P$ does not uniquely determine $P,Q$, even when assuming $Q$ is real.  Nonetheless, we only need to find one set of phase factors $\Phi^*$ to accurately encode $f(x)$. 

Our optimization problem can be viewed as variational quantum circuit (more specifically, similar to the quantum approximate optimization algorithm (QAOA) \cite{farhi2014quantum}),  in which one set of quantum gates (those associated with $\sigma_x$) are fixed. Due to the complex energy landscape, a good initial guess is necessary for the performance of the optimizer.

\subsection{Generating approximation polynomials}\label{subsec:approxpolyn}
In order to generate a polynomial to approximate $f$ to a given degree, we consider in this work two efficient approaches: the Fourier-Chebyshev expansion method and the Remez  method.

For a real smooth function $F$ on the interval $[-1, 1]$, we find its polynomial approximation in terms of Chebyshev polynomial of the first kind, \ie, $F(x) \approx f(x) = \sum_{j=0}^d c_j T_j(x)$. The Fourier approach uses the fast Fourier transformation (FFT) to efficiently evaluate the coefficients via a quadrature 
 \begin{equation}
    c_{j} \approx \frac{\left(2-\delta_{j 0}\right)}{2 K}(-1)^{j} \sum_{l=0}^{2 K-1} F\left(-\cos \theta_{l}\right) e^{\I j \theta_{l}} 
\end{equation}
where $\theta_{l}=\pi l / K, 0 \leq l \leq 2 K-1$, and $K$ is the number of quadrature points. 

We may alternatively consider optimization with respect to the $L^\infty$ norm. In fact, we may even restrict the interval of approximation to be a subset $[a,b]\in[-1, 1]$. In this case, an approximation polynomial can be obtained by solving the optimal approximation problem in terms of the $L^\infty$ norm
\begin{equation}\label{eq:optapproxpart}
    f=\argmin{f\in \mathbb{R}[x], \deg(f)\le \qspdeg}{\max_{x\in[a,b]}|F(x)-f(x)|}.
\end{equation}
The Remez algorithm \cite{Remez1934,cheney1966introduction} allows efficient solution of \cref{eq:optapproxpart}. This is an iterative method consisting of two steps. In the first step, we find the coefficients of $f$ from $\qspdeg+2$ points sampled from the interval by solving a set of linear equations. The second step involves adjusting $\qspdeg+2$ samples from coefficients solved in the first step.  We can also use the Remez algorithm to solve for $f$ using parity constraint. Full details are given in \cref{sec:remez}. 

\subsection{Choice of initial point}
The objective function in the optimization model of \cref{eq:alg-prac} is highly non-convex, rendering the global minimum hard-to-find. Numerical tests given in \cref{app:initial} illustrate that the solver can easily get stuck in a local minimum if we initiate it randomly, confirming the complexity of the landscape.
 Another possible choice of the initial phase factors is $\Phi=(0,0,\dots,0,0)$. Then the components of QSP matrix are Chebyshev polynomials $P(x) = T_\qspdeg(x)$ and $Q(x) = \seccheby_{\qspdeg-1}(x)$. However, straightforward computation shows that in this case we have $\nabla L(\hat{\Phi})=0$, \ie, $\hat{\Phi}$ is a stationary point, and obviously $L(\hat{\Phi})\ne 0$. 

Our main observation is that if we slightly modify the initial point as 
\begin{equation}
\Phi=\left(\frac{\pi}{4},0\dots,0,\frac{\pi}{4}\right) \in \RR^{\qspdeg+1},
\label{eqn:phi_initial}
\end{equation}
or correspondingly, the symmetrized version
\begin{equation}
\hat{\Phi}^0=\left(\frac{\pi}{4},0\dots,0\right)\in \RR^{\tilde{\qspdeg}},
\label{eqn:phi_initial_sym}
\end{equation}
then a gradient-based algorithm can reach a global minimum in all cases shown in \cref{sec:num}. According to the discussion in \cref{sec:symmetry}, this corresponds to the initial guess with $P(x)=\I T_\qspdeg(x)$ and $Q(x)=\seccheby_{\qspdeg-1}(x)$. The intuitive reason for choosing such an initial point is that we are interested in the real part of $P(x)$. The choice in \cref{eqn:phi_initial} ensures that $\Re[P(x)]=0$, which is unbiased with respect to the function to be approximated. On the other hand, the seemingly natural choice $\Phi=(0,0,\dots,0,0)$ gives  $P(x)=T_\qspdeg(x)$, which is a heavily biased initial guess of the real component. The theoretical study of the landscape around such an initial guess justifying the effectiveness of such a choice of the initial guess will be the focus of future work.

\subsection{Algorithm}\label{sec:alg}

We use a quasi-Newton method to perform numerical optimization of the phase factors. Compared to the Newton type method, we find that a quasi-Newton method such as the L-BFGS method \cite[Chapter~5]{sun2006optimization} leads to fast convergence without any need to evaluate the Hessian matrix, for which the computational cost would scale as $\Or(\qspdeg^3)$. \cref{sec:l-bfgs} describes the L-BFGS algorithm, which is applied to the symmetry-reduced phase factors according to \cref{eq:alg-prac}. Using the initial phase factors in \cref{eqn:phi_initial_sym}, the Hessian matrix $\text{Hess}\, L(\hat{\Phi}^0)$ is a constant matrix regardless of approximation polynomial $f$. More specifically, we have 
\begin{equation}\label{eq:alg-hess}
 \text{Hess}\, L(\hat{\Phi}^0)=\left\{\begin{array}{ll}
    2I  & \qspdeg\text{ is odd,}  \\
    \text{diag}(2,\dots,2,1)  &  \qspdeg\text{ is even.}
 \end{array}    \right.
\end{equation}
The inverse of this Hessian matrix will be fed into the L-BFGS algorithm.
In \cref{alg:bfgs} below we describe how to compute optimal phase factors corresponding to a given polynomial. The complete procedure to approximate a generic complex-valued function as polynomial components is presented in \cref{alg:sum}. 

\begin{algorithm}[htbp]
\caption{\textbf{Function:} $\hat{\Phi}=\text{QSPBFGS}(\hat{\Phi}^0,f,\epsilon)$} 
\label{alg:bfgs}
\begin{algorithmic} 
\STATE \textbf{Input:} An initial vector $\hat{\Phi}^0$, a real polynomial $f$ of degree $\qspdeg$ and error tolerance $\epsilon$.
\vspace{1em}
\STATE Choose $\tilde \qspdeg=\lceil \frac{\qspdeg+1}{2}\rceil$ points $x_j=\cos(\frac{(2j-1)\pi}{4\tilde \qspdeg})$ as the positive roots of Chebyshev polynomial $T_{2\tilde \qspdeg}$.
\STATE  {Construct objective function $L(\hat{\Phi})$ using  \cref{eq:alg-prac}. \setlength{\belowdisplayskip}{0pt} \setlength{\belowdisplayshortskip}{0pt}
\setlength{\abovedisplayskip}{0pt} \setlength{\abovedisplayshortskip}{0pt}}
\STATE Choose the initial approximation of inverse Hessian $B_0$ using \cref{eq:alg-hess}.
\STATE Set $t=0$
\WHILE{$L(\hat{\Phi})>\epsilon$}
    \STATE Obtain $\hat{\Phi}^{t+1}$ by updating $\hat{\Phi}^{t}$ via L-BFGS algorithm.
    \STATE Set $t=t+1$.
\ENDWHILE
\STATE \textbf{Return:} $\hat{\Phi}^t$
\end{algorithmic}
\end{algorithm}

\begin{algorithm}[htbp] 
\caption{Finding phase factors for the polynomial approximation of a smooth function $f$ over interval $[a,b]$} 
\label{alg:sum}
\begin{algorithmic} 
\STATE \textbf{Input:} A complex-valued function $F\in C^\infty [a,b]$, a non-negative integer $\qspdeg$ and error tolerance $\epsilon$.
\vspace{1em}
\STATE Find polynomial $f\in \mathbb{C}[x]$ of degree at most $n$ which approximates $f$ over the interval $[a,b]$. One can obtain such polynomial via the Fourier-Chebyshev expansion approach or the Remez algorithm \cite{Remez1934,cheney1966introduction}.
\STATE Scale $f$ by a constant factor $\alpha$.
\STATE Denote $f_j, j = 1, 2, 3, 4$ as real/imaginary and even/odd part of $f/\alpha$.
\STATE Set $\hat{\Phi}^0=(\frac{\pi}{4},0,\dots,0)\in \RR^{\tilde \qspdeg}$.
\STATE Solve $\hat{\Phi}_j=\text{QSPBFGS}(\hat{\Phi}^0,f_j,\epsilon)$ for each component.
\STATE \textbf{Return:} $\hat{\Phi}_j, j = 1, 2, 3, 4$ and factor $\alpha$.
\end{algorithmic}
\end{algorithm}

\section{Numerical results}\label{sec:num}

We present a number of tests to examine the effectiveness of the optimization based method compared to the previous direct methods. We implement the direct algorithms designed in \cite{GilyenSuLowEtAl2019} and \cite{Haah2019} (denoted here as the GSLW and Haah methods, respectively). All numerical tests are performed on an Intel Core 4 Quad CPU at 2.30 GHZ with 8 GB of RAM. Our method is implemented in \textsf{MATLAB} R2018b, while the GSLW and the Haah method are written in  \textsf{Julia} 1.2 for its better support for high-precision arithmetic. Our implementation (optimization, GSLW, Haah) can be downloaded from the Github repository\footnote{\url{https://github.com/qsppack/QSPPACK}}.

We utilize the \textsf{BigFloat} type to achieve variable precision arithmetic and internal routines in  \textsf{Julia} for the root-finding procedures. In \cref{app:implement}, we present the details of algorithms used for comparison and state some modifications to enhance the numerical stability. The stopping criterion is
\begin{equation}
    \max_{j=1,\dots,\tilde \qspdeg}\left|\qspobj{\Phi}{x_j} - f(x_j)\right|<\epsilon
\end{equation}
for both the GSLW method and our optimization method. The Haah method is terminated when the resulting factors are $\epsilon$-close to the target polynomial of degree $\qspdeg$ for values on the $\qspdeg$-th roots of unity. We set $\epsilon$ to be $10^{-12}$. We highlight the critical feature 
that all of the arithmetic in our optimization algorithm is performed using only double-precision floating-point numbers. This is a remarkable advantage in terms of computation cost and numerical stability compared to the direct algorithms, which have to make use of variable precision arithmetic operations. In fact, our numerical results indicate that even with variable precision arithmetic operations, both the GSLW and the Haah method still struggle to find the phase factors accurately when the degree of polynomial becomes large ($\gtrsim 500$).

\subsection{Hamiltonian Simulation}\label{subsec:HamiltonianSimulation}
A Hermitian matrix $H$ with bounded norm $\| H \|_2 \leq 1$ has the spectral decomposition $H = \sum_j \lambda_j | j \rangle\langle j |$. The Hamiltonian simulation with duration $\tau$ through $H$ is then given by $f(H) = e^{- \I \tau H} = \sum_j e^{- \I \tau \lambda_j} | j \rangle\langle j |$. Implementation of Hamiltonian simulation is thus determined by the phase factors that approximate the smooth complex-valued function $f(x) = e^{- \I \tau x}$. Since this is smooth on the interval $[-1, 1]$, its polynomial approximation can be generated from the Jacobi-Anger expansion\cite{BerryChildsKothari2015}:
\begin{equation}\label{eq:num-jacobi}
\begin{aligned}
    e^{-\I\tau x} &= J_0(\tau)+2\sum_{k\,\,\mathrm{even}} (-1)^{k/2}J_{k}(\tau)T_k(x)\\&+2\I \sum_{k\,\,\mathrm{odd}} (-1)^{(k-1)/2}J_{k}(\tau) T_k(x).
\end{aligned}
\end{equation}
Here  $J_k$'s are the Bessel functions of the first kind. The $L^\infty$ error to truncate the series up to order $\qspdeg$ is bounded by
\begin{equation}
\begin{aligned}
    &2 \sum_{k = \qspdeg+1}^\infty \left|J_{k}(\tau)\right| \leq 2 \sum_{k = \qspdeg+1}^\infty \left( \frac{e |\tau|}{2} \right)^k k^{-k}\\
    &\lesssim e^{-\qspdeg} \sum_{k = \qspdeg+1}^\infty \frac{1}{k!} \left( \frac{e |\tau|}{2} \right)^k < e^{e |\tau| / 2 - \qspdeg}.
\end{aligned}
\end{equation}
Thus, the truncated series up to $\qspdeg \approx e |\tau| / 2 + \log( 1/\epsilon_0)$ leads to an approximation whose truncation error is bounded by $\epsilon_0$. In our simulation, we simply choose $\qspdeg=1.4|\tau|+\log(1/\epsilon_0)$, where $\epsilon_0=10^{-14}$, to make the truncation error negligible compared to the error caused by other factors. We denote such an approximation for Hamiltonian simulation with duration $\tau$ by $f_\tau$.

We compare our method with the GSLW and Haah methods on the polynomial given by \cref{eq:num-jacobi}. For each $\tau$, we divide $f_\tau$ into real and imaginary parts, and perform algorithms separately according to case 3 in \cref{sec:matrixpolynomial}. Then, we sum up the CPU time and the error together of each part as final results. We divide the coefficients of $f_\tau$ by a constant factor $2$ to ensure $|f_\tau|\le 1$ for $x\in[-1,1]$. The CPU time and the number of bits utilized to perform arithmetic are displayed in \cref{fig:HStime} and \cref{fig:HSbits}, respectively,
together with polynomial fits to the data for large $\tau$ values in \cref{fig:HStime} (the points for small $\tau$ values are in the pre-asymptotic regime and are excluded in the fits).

We display results for $\tau$ up to $500$ since the direct methods become very inefficient for larger values of $\tau$. In particular, the GSLW method fails to yield phase factors with required accuracy $\epsilon=10^{-12}$ when the degree $\qspdeg$ of $f_\tau$ is larger than $369$. We contribute the failure to the instability of \textsf{Julia}'s internal root-finding procedure. We observe that the CPU time of our proposed method scales as $\tau^2$, while it scales as $\tau^3$ for the Haah method. Moreover, for both the GSLW and the Haah method the number of bits required is linear in $\tau$, while our optimization method is seen to be numerically stable in all calculations with use of only standard double precision arithmetic operations, \ie, the number of bits is independent of $\tau$.  

\begin{figure}[htbp]
    \centering
    \vspace{-10pt}
    \subfigure[\label{fig:HStime}]{
        \includegraphics[width=8cm]{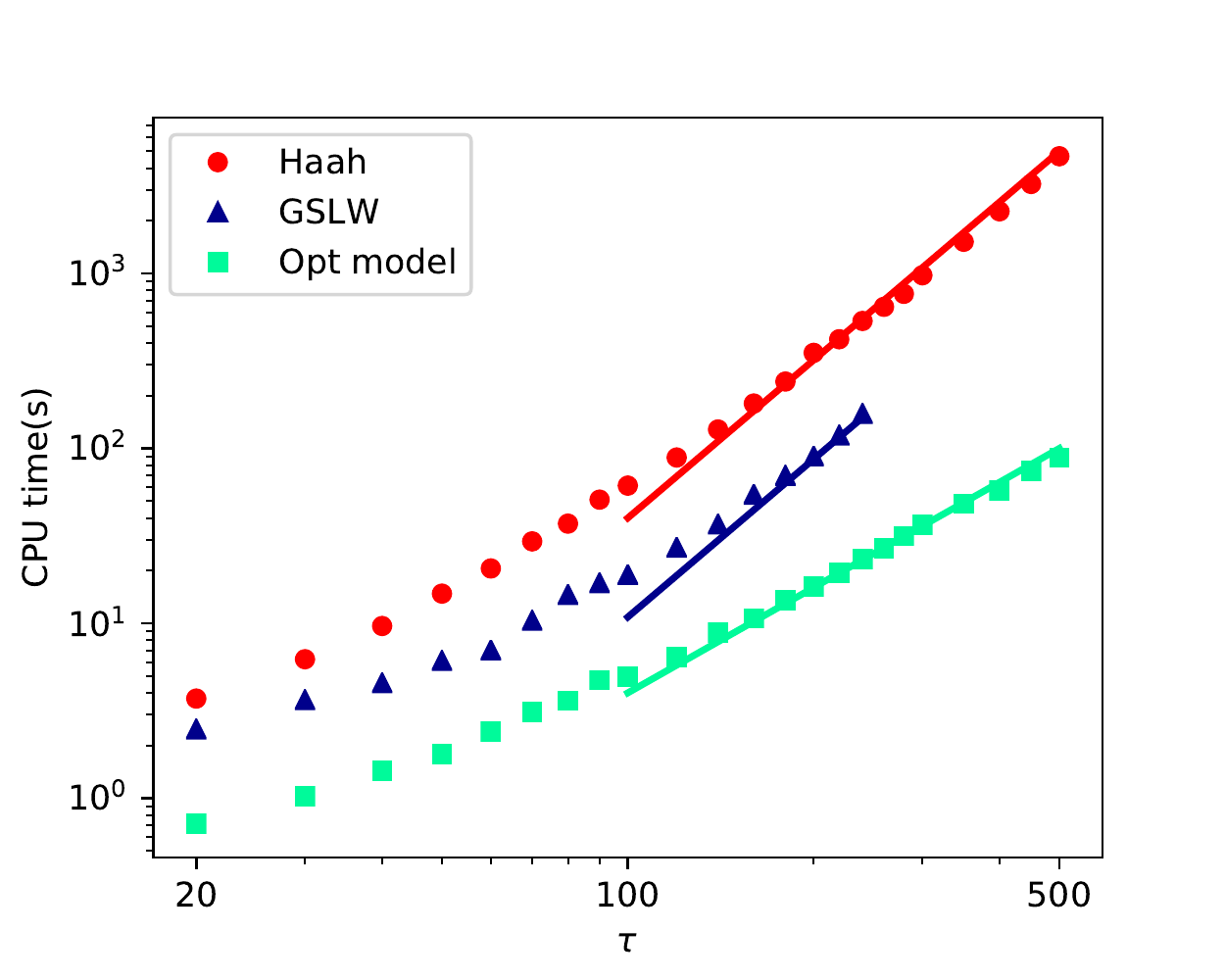}
    }
    
    \subfigure[\label{fig:HSbits}]{
        \includegraphics[width=8cm]{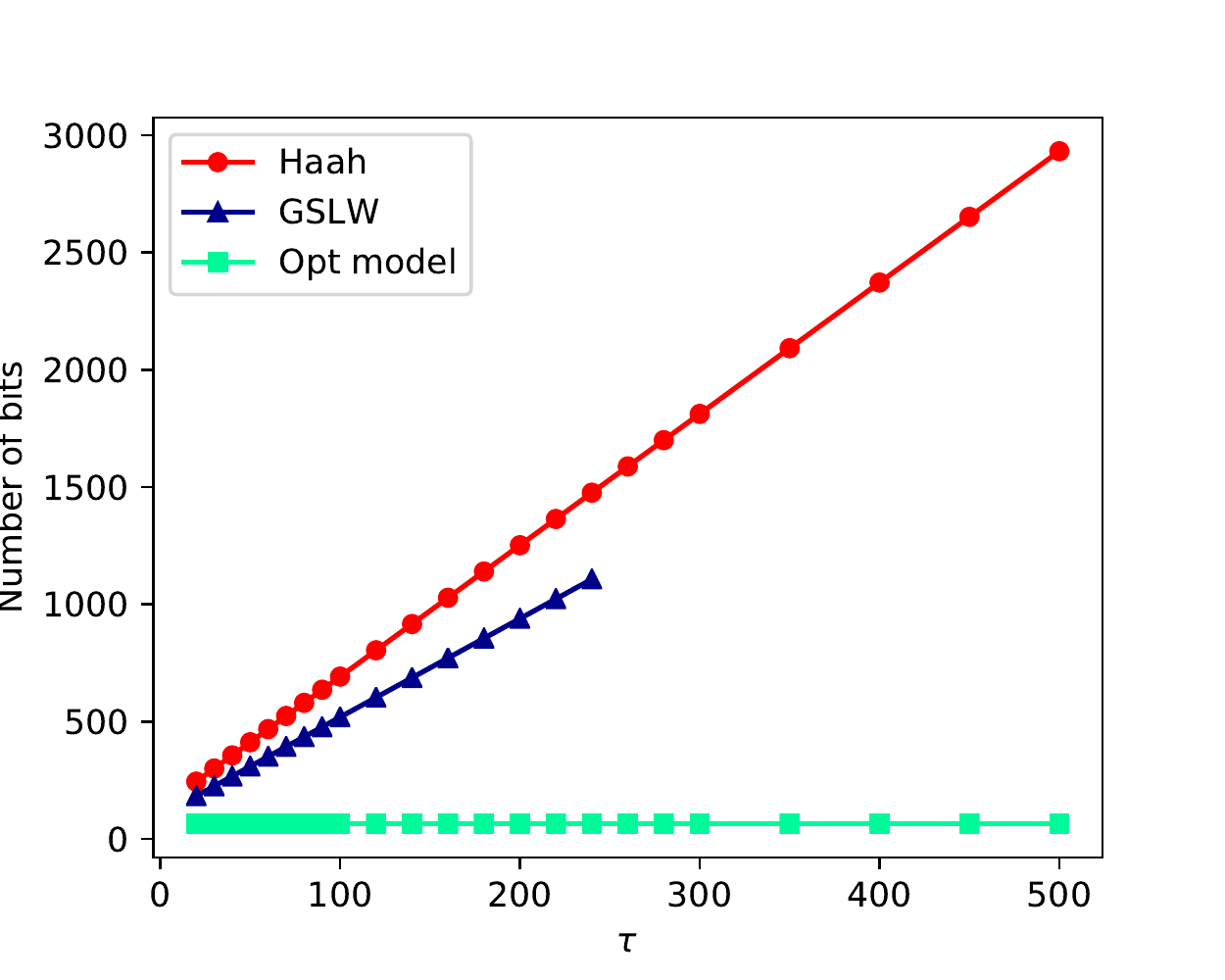}
        }
        \vspace{-10pt}
    \caption{Resource costs in determining QSP phase factors for the Hamiltonian simulation problem. Red dots, blue triangles and green squares correspond to the results by using Haah, GSLW and our optimization method, respectively. (a) CPU time(s) spent by each algorithm as a function of duration $\tau$, together with polynomial fits in the large $\tau$ region. The degree of polynomial is $\qspdeg=1.4|\tau|+\log(1/\epsilon_0)$, with $\epsilon_0=10^{-14}$. The slope of the red (gray) and the blue (dark gray) lines is 3, representing $ \text{CPU time}= \text{const} \cdot \tau^3$. The slope of the green (light gray) line is 2,  representing $ \text{CPU time}= \text{const} \cdot \tau^2$. (b) Number of bits used to store floating-point numbers and perform arithmetic. We show results for the GSLW method only up to $\tau=240$ since it fails to generate accurate  phase factors for larger $\tau$.}
     \vspace{-10pt}
\end{figure}

To further demonstrate the capability of our method, we test our algorithm with $\tau$ up to $5000$. When $\tau=5000$, the polynomial degree $d$ is $7033$. The computational cost for evaluating the real and imaginary parts of $f_\tau$ is given in \cref{fig:hsfull}. We also display in \cref{tab:hs} the $L^\infty$ error (\ie \ the maximum error) between the polynomial given by QSP phase factors and $e^{- \I \tau x}$, to verify the robustness of our method and the effectiveness of our choice of the stopping criterion. The CPU time still scales asymptotically as $\tau^2$, in agreement with our expectations since the  per-iteration cost of the optimization procedure is $\mathcal{O}(\qspdeg^2)$.

\begin{figure}[htbp]
    \centering
    \vspace{-10pt}
    \includegraphics[width=8cm]{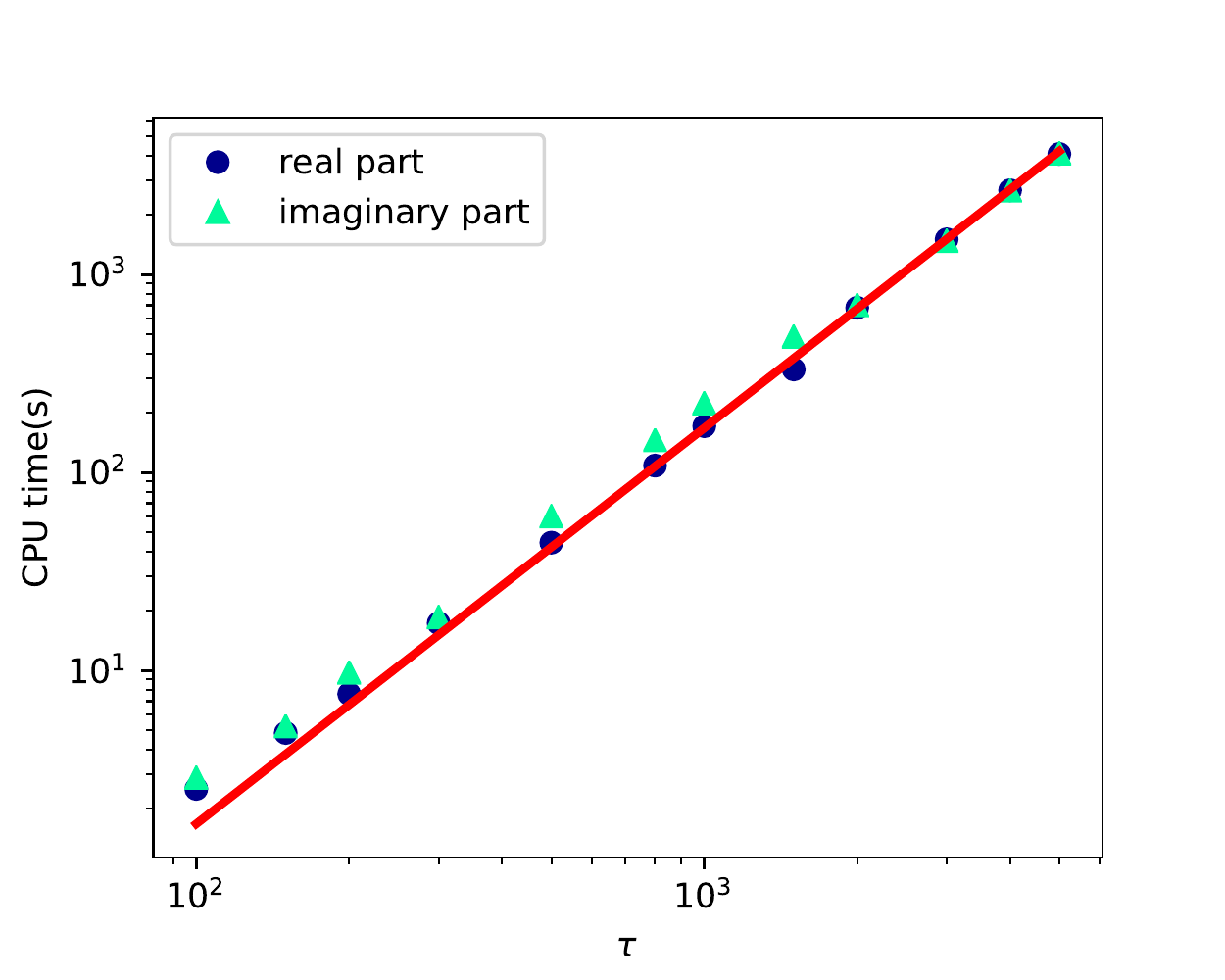}
    \caption{CPU time(s) required using the optimization algorithm for determining QSP phase factors for Hamiltonian simulation, shown as a function of $\tau$. Blue dots (green triangles) correspond to the real (imaginary) part of $f_\tau$ of degree $\qspdeg=1.4|\tau|+\log(1/\epsilon_0)$ with $\epsilon_0=10^{-14}$. The slope of the red line is $2$, representing $ \text{CPU time}= \text{const} \cdot \tau^2$.}
    \label{fig:hsfull}
\end{figure}

\begin{table}[htbp]
    \centering
    \begin{tabular}{|c|cccccc|}\hline
      $\tau$ & 100 & 150 &200 & 300& 500 & 800\\ \hline
      real &6.1e-13 &7.9e-13& 1.1e-12&  2.4e-13& 4.7e-13 & 3.6e-13 \\
      imaginary & 1.1e-12 & 2.3e-13& 3.3e-13& 3.2e-13& 2.8e-13 & 5.9e-13 \\ \hline
    $\tau$ & 1000& 1500& 2000& 3000& 4000& 5000 \\ \hline
    real& 5.6e-13 &5.5e-13 &5.5e-13& 7.2e-13 &1.2e-12& 9.4e-13 \\ 
    imaginary & 4.2e-13&5.9e-13& 9.0e-13& 7.3e-13& 9.0e-13& 1.5e-12  \\ \hline
    \end{tabular}
    \caption{$L^\infty$ error of the optimization algorithm for determining QSP phase factors for Hamiltonian simulation as a function of $\tau$. The degree of truncated polynomial is $\qspdeg=1.4|\tau|+\log(1/\epsilon_0)$ with $\epsilon_0=10^{-14}$.}
    \label{tab:hs}
\end{table}

\subsection{Eigenstate filtering function}\label{subsec:EigenstateFilter}
In order to prepare an eigenstate corresponding to a known eigenvalue,  we consider the following $2k$-degree polynomial
\begin{equation}\label{eq:num-eigen}
    f_k(x,\Delta)=\frac{T_k(-1+2\frac{x^2-\Delta^2}{1-\Delta^2})}{T_k(-1+2\frac{-\Delta^2}{1-\Delta^2})}.
\end{equation}
 Suppose $H$ is a Hermitian matrix with an eigenvalue $\lambda$ that is separated from other eigenvalues by a gap $\Delta>0$. Let $\tilde H=(H-\lambda I)/(\alpha+|\lambda|)$ and $\tilde \Delta=\frac{\Delta}{2\alpha}$. It was proven in \cite{LinTong2019} that
\begin{equation}\label{eq:num-eig}
    \|f_k(\tilde H,\tilde \Delta)-\hat{P}_\lambda\|_2 \le 2e^{-\sqrt 2k\tilde \Delta},
\end{equation}
where $\hat{P}_\lambda$ is the projection operator onto the eigenspace corresponding to $\lambda$. Furthermore, $f_k$, which is referred to as the eigenstate filtering function,  is the optimal polynomial for filtering out the unwanted information from all other eigenstates.

For this demonstration we assume $\lambda=0$, and $\alpha=1$. We choose $\Delta=0.1,0.05,0.01, 0.005$ and test our algorithm with different target filter values $k$. \cref{eq:num-eig} indicates that $k\Delta$ controls the accuracy of the approximation. For each $\Delta$ we choose $k$ such that $k\Delta=3,5,10,15,20,25$, respectively. The largest polynomial in this example is $d=10,000$. The coefficients of  polynomials are divided by $\sqrt 2$ to avoid instabilities during optimization (see \cref{sec:sensitivity} for reasons to scale the function). The results are summarized in \cref{fig:Eigentime} and \cref{tab:Eigenerr}. From the figure we observe that the optimization method performs stably in all cases, with CPU time scaling as $k^2$. These results are compared with the corresponding results for the direct methods of GSLW and Haah in \cref{fig:Eigencom}, for $\Delta$ ranging from $0.005$ to $0.1$. This comparison is made only for $k\Delta=3$, since we observe that direct methods struggle to treat larger values of $k\Delta$.  It is evident that the optimization algorithm also shows superior performance to the direct methods in this example. 

In particular, the Haah method fails to solve the QSP phase factors with required accuracy $\epsilon=10^{-12}$ when $\Delta$ is less than $0.01$. The weaker performance of the Haah method compared to (our modified) GSLW method observed in \cref{fig:Eigencom} can be attributed to the following reasons. We note that \textsf{Julia}'s internal root-finding routine has difficulty finding all the roots of a polynomial when its degree is high, even when variable precision arithmetic operations are used. The performance of the GSLW and Haah methods can thus depend on the dataset, since they apply the root-finding procedure to different polynomials. We observe that sometimes the GSLW method can reach a  polynomial of higher degree than the Haah method, and sometimes it is the other way around. We remark that the degree of polynomial fed into the Haah method is twice as large as that fed into the GSLW method, since  the variable $x$ is replaced by  $(z+1/z)/2$ in the Haah method. This increases the difficulty for the Haah method to solve phase factors successfully. By contrast, our modified implementation of the GSLW method (Appendix \ref{app:implement}) expands the polynomial in the Chebyshev basis, which significantly increases its stability, making its performance comparable to Haah's.

\begin{figure}[htbp]
    \centering
    \includegraphics[width=8cm]{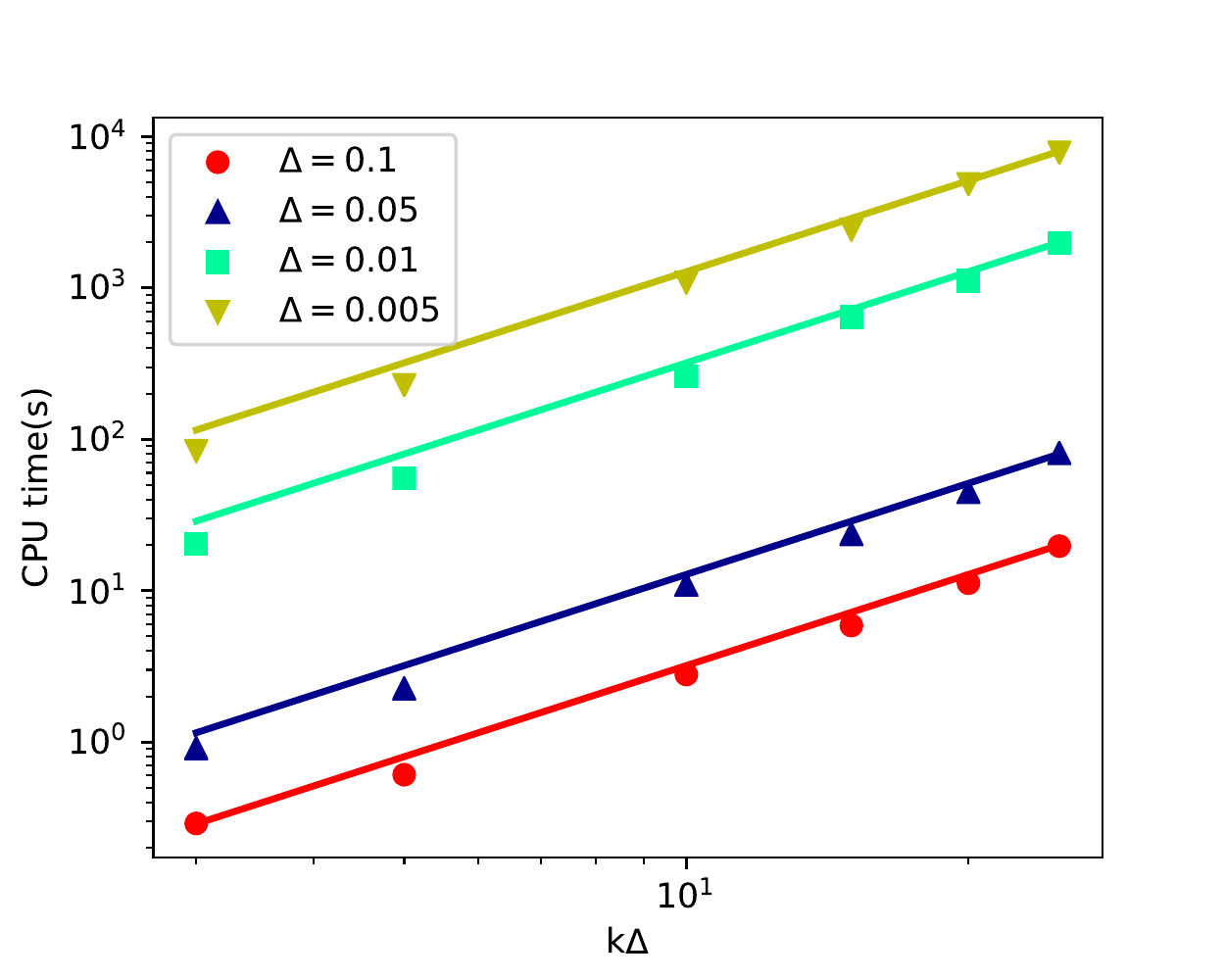}
    \caption{CPU time(s) using the optimization algorithm for determining the set of phase factors for the eigenstate filter, shown as a function of $k\Delta$. The degree of each polynomial is $2k$. The slope of each line is $2$, reflecting the quadratic cost $ \text{CPU time}= \text{const} \cdot k^2$. }
    \label{fig:Eigentime}
\end{figure}

\begin{table}[htbp]
    \centering
    \begin{tabular}{|c|cccccc|}\hline
    \diagbox{$\Delta$}{$k\Delta$} & 3 & 5 & 10 & 15 & 20 & 25 \\ \hline 
    0.1   &  3.4e-14 &5.2e-13 &1.1e-13 &1.1e-12 &8.9e-14 &8.5e-13 \\
    0.05  & 3.2e-14 & 4.9e-13 & 1.1e-13 & 1.1e-12 &1.0e-13 & 8.4e-13 \\
    0.01 & 4.7e-14 & 4.9e-13 & 1.7e-13 & 1.1e-12 & 2.2e-13& 8.1e-13 \\
    0.005 & 2.1e-13 & 5.6e-13 & 2.1e-13 & 1.2e-12 & 4.7e-13 & 8.8e-13 \\ \hline
    \end{tabular}
    \caption{$L^\infty$ error of the optimization algorithm for determining the QSP phases for the eigenstate filter defined in \cref{eq:num-eigen}.}
    \label{tab:Eigenerr}
\end{table}

\begin{figure}[htbp]
    \centering
    \includegraphics[width=8cm]{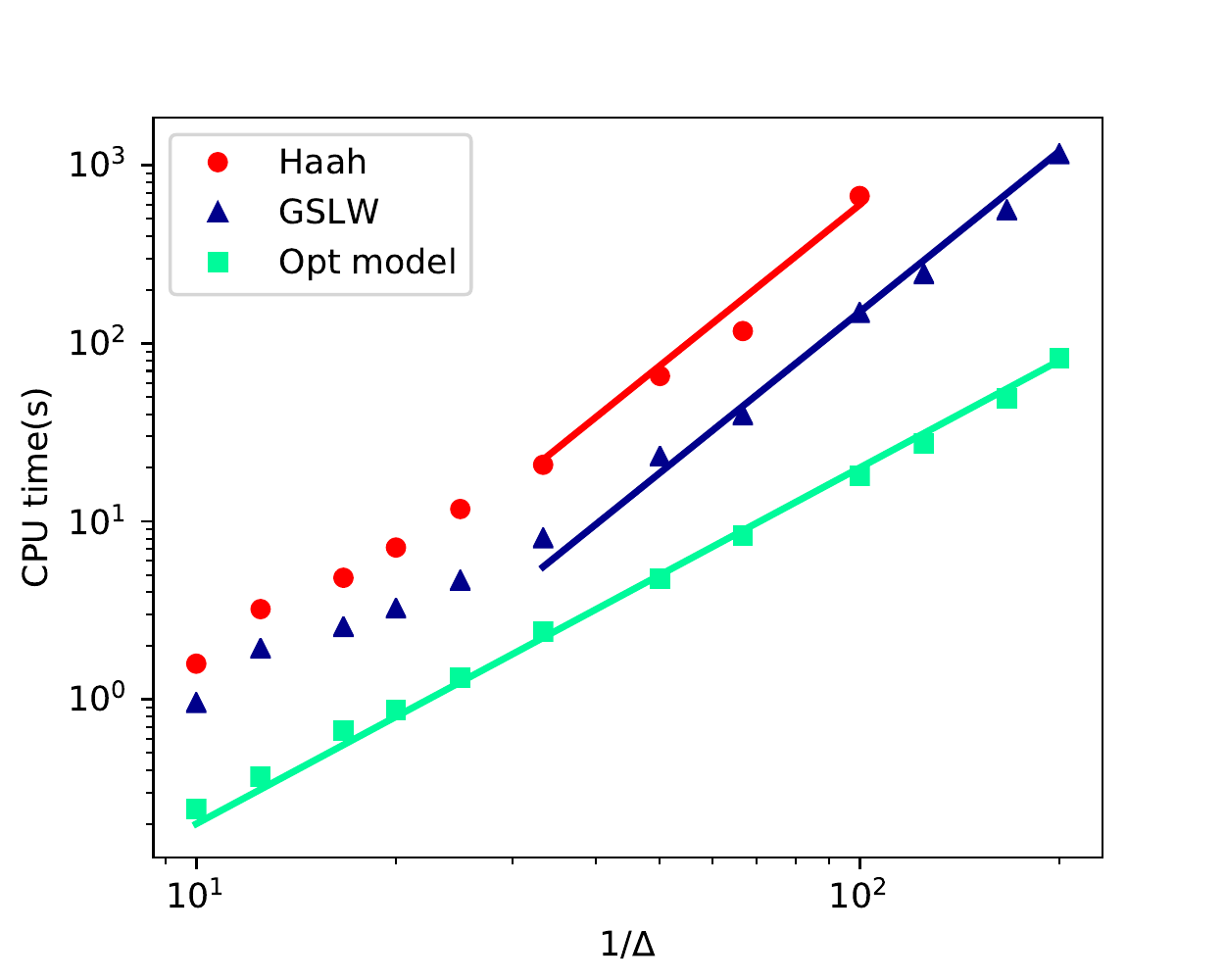}
    \caption{CPU time(s) of the optimization algorithm (green squares) compared with direct (GSLW and Haah) algorithms (blue triangles and red dots, respectively) for determining QSP phases for the eigenstate filter, shown as a function of $1/\Delta$. The degree of each polynomial is $2k=6/\Delta$. The slopes of the red (gray) and blue (dark gray) lines are 3, corresponding to a cubic cost in $\tau$.
    The slope of the green (light gray) line is 2,  corresponding to a quadratic cost in $\tau$.}
    \label{fig:Eigencom}
\end{figure}

\subsection{Matrix inversion}\label{subsec:MatrixInverse}
Consider the quantum linear problem $A|x\rangle=|b\rangle$ where $A$ is a Hermitian  matrix whose condition number is $\kappa$. Then the eigenvalues of $A$ are distributed within the interval  $D_\kappa := [-1,-1/\kappa]\cup[1/\kappa, 1]$. The solution $|x\rangle$ can be constructed via matrix inversion, using QSP to generate the action of $A^{-1}$. For this we need a polynomial approximation of $1/x$ on the interval $D_\kappa$. We consider two options here.  The first is to generate a polynomial approximation of $1/x$ on $D_\kappa$ by extending the function to the interval $[-1, 1]$ via an approximate function, as outlined in \cref{sec:opt-cho} above. The second is to apply the Remez algorithm \cite{Remez1934,cheney1966introduction} directly to the interval $D_\kappa$.  The first approach was pursued in \cite{ChildsKothariSomma2017}, where the following odd extension was proposed
\begin{equation}
    g(x):=\frac{1-(1-x^2)^b}{x}.
\end{equation}
Then, the truncated sum of Chebyshev polynomials
\begin{equation}\label{eq:num-x}
    f(x)=4\sum_{j=0}^\qspdeg (-1)^j\frac{\sum_{i=j+1}^{b}{\binom{2b}{b+i}}}{2^{2b}}T_{2j+1}(x)
\end{equation}
is $\epsilon_0$-close to $1/x$ on $D_\kappa$ by choosing $b=\left\lceil \kappa^2\log(\frac{\kappa}{\epsilon_0}) \right\rceil$ and $\qspdeg=\left\lceil \sqrt{b\log(\frac{4b}{\epsilon_0})}\right\rceil$. In the test made here $\epsilon_0$ is set to be $10^{-14}$.

In the second approach using the Remez algorithm, our goal for the matrix inversion problem is to directly construct an odd polynomial that approximates $f(x)=1/x$ on $D_{\kappa}$.  More generally, we note that 
if $A$ is positive definite and  $D_\kappa = [1/\kappa, 1]$, then we may approximate $f$ by extending it to a function that is either even or odd.  Since this paper focuses on the problem of finding the phase factors for approximating a smooth function in general, we will consider both the even and odd extensions below. 
For the current instance $f(x) = 1/x$, we gradually increase the degree $\qspdeg$ until the value of $f(x)$ obtained by the Remez algorithm approximates $1/x$ over $D_\kappa$ with $L^\infty$ error below $\epsilon_0$.  Fig.\ref{fig:Remez} compares the polynomial given by the Fourier-Chebyshev method, \cref{eq:num-x}, with that generated by the Remez method, for $\kappa=20$ and $\epsilon_0=10^{-3}$.  

\begin{figure}[htbp]
    \centering
    \includegraphics[width=8cm]{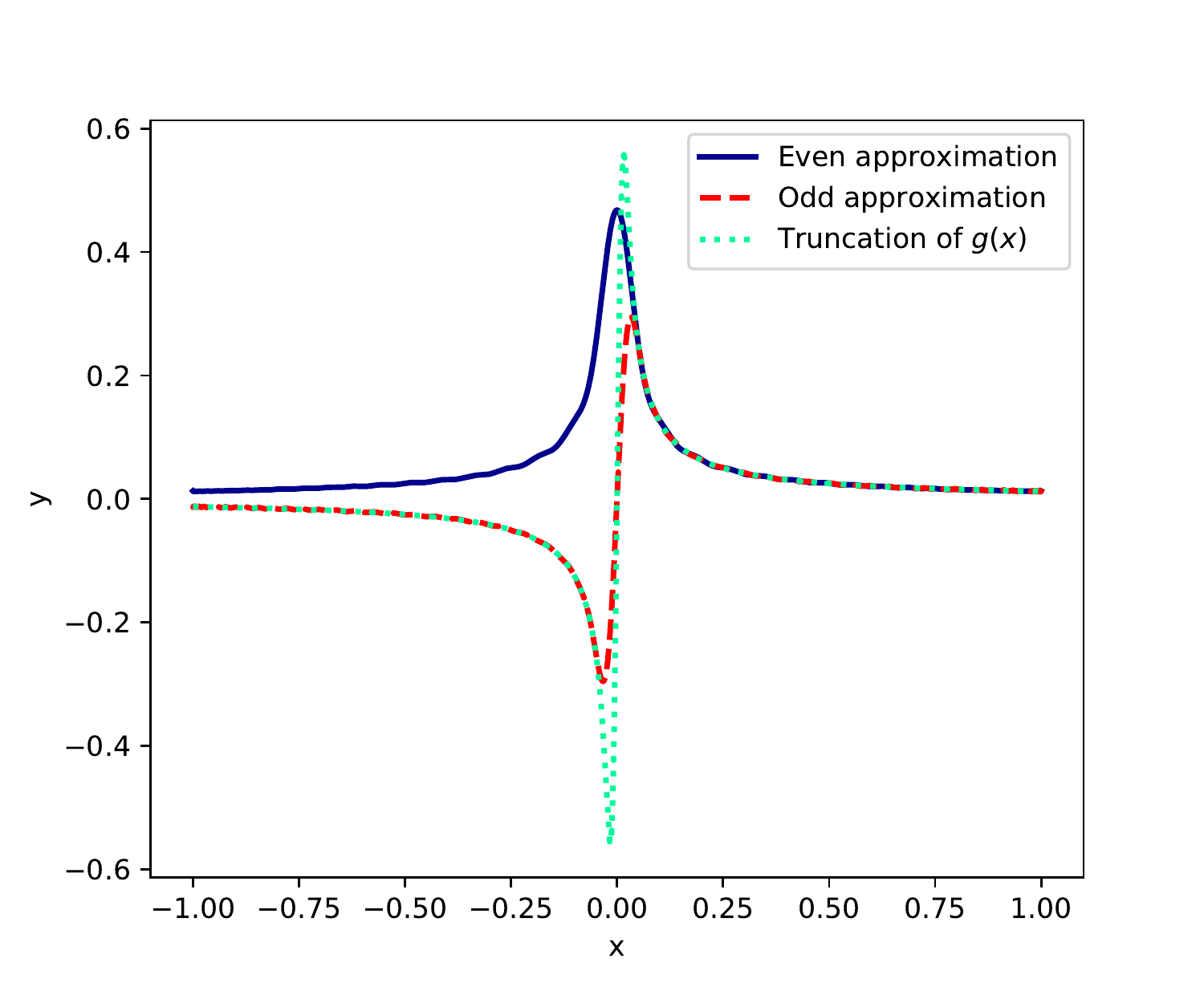}
    \caption{Comparison of the form of polynomials given by the Fourier-Chebyshev method, \cref{eq:num-x}, with those generated by the Remez method, for odd and even parities. The degree of the truncated polynomial here is $611$ and degrees of the even (odd) approximation polynomials generated by the Remez method are 76 (111). The approximation polynomials are divided by $80$ for this plot.}
    \label{fig:Remez}
\end{figure}

In this example we choose $\kappa=10,20,\dots,50$. We test our algorithm with $\epsilon_0=10^{-14}$ on polynomials given by \cref{eq:num-x} and generated by the Remez algorithm with odd and even parity, respectively. The CPU time associated with each polynomial approximation is presented in \cref{fig:xinv}. We also compare the optimization method with the GSLW and the Haah method on the polynomials with lower degrees. We choose $\epsilon_0=10^{-6}$ and generate polynomials by the Remez algorithm with odd and even parity. The results of the comparison are demonstrated in \cref{fig:xcom}, while the degrees of the polynomials given by each method are shown in \cref{tab:xinv}. Similar to the case of eigenstate filtering polynomials, we find that the GSLW and Haah methods cannot reach the target accuracy when the degree of the polynomial becomes large. Hence we reduce the accuracy in order to decrease the polynomial degrees here.

\cref{tab:xinv} indicates that use of the Remez method can significantly reduce the degree of polynomials needed to approximate $1/x$, with a reduction of to a factor of $2\sim 3$. We find that the even polynomial approximation is slightly less expensive than the odd expansion. This is due to the fact that an even extension has smaller gradient near the origin, compared with that of the odd extension, as shown in \cref{fig:Remez}.  Our proposed optimization method performs well on these examples, yielding phase factors robustly, with computational cost scaling quadratically with respect to $\kappa$. The largest polynomial degree $d=4035$. 
\begin{table}[htbp]
    \centering
    \begin{tabular}{|c|ccccc|}\hline
$\kappa$ & 10 & 20 & 30 & 40 &50\\ \hline
Truncation of $g(x)$ ($\epsilon_0=10^{-14}$) & 759 & 1559 & 2375 & 3201 & 4035\\ \hline
Odd Remez ($\epsilon_0=10^{-14}$)& 303 & 607 & 911 & 1215 & 1519\\ \hline
Even Remez ($\epsilon_0=10^{-14}$)& 280 & 560 & 840 & 1020 &1400\\ \hline
Odd Remez ($\epsilon_0=10^{-6}$)& 125 & 249 & 373 & 499 & 623 \\ \hline
Even Remez ($\epsilon_0=10^{-6}$)& 104 & 206 & 310 & 412 & 516\\ \hline
    \end{tabular}
    \caption{Degrees of approximation polynomials with accuracy $\epsilon_0$ given by an odd smearing function in \cref{eq:num-x} and the Remez method with odd and even parity, respectively.}
    \label{tab:xinv}
\end{table}

\begin{figure}[htbp]
    \centering
    \includegraphics[width=8cm]{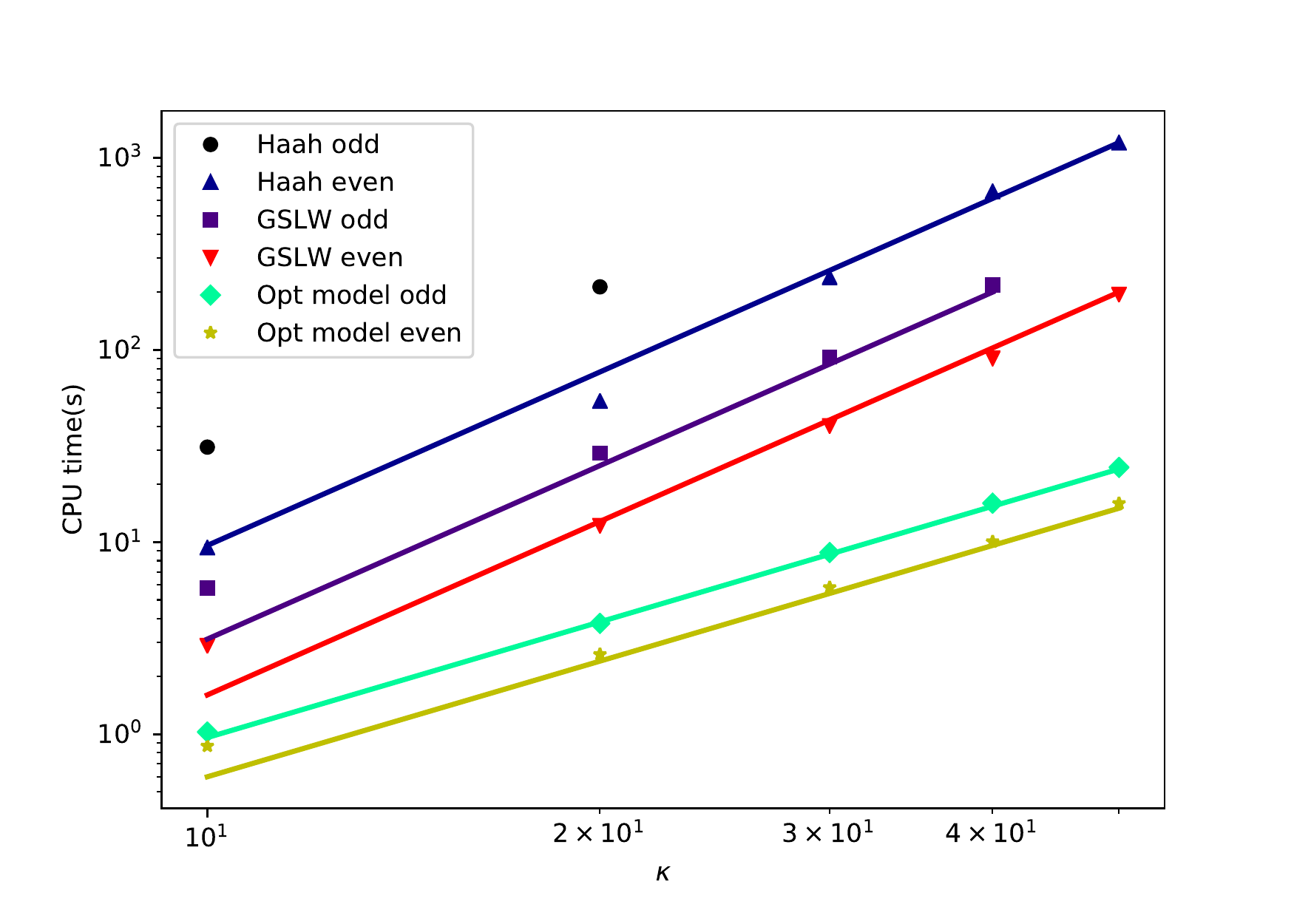}
    \caption{CPU times for approximating $1/x$ over $D_\kappa$ via QSP as a function of $\kappa$ for the optimization method, compared with the corresponding times for the GSLW and Haah methods. Lines labelled  ``even'' (``odd'') represent the results of approximating polynomials given by the Remez method with even (odd) parity. The slopes of the two lowest lines are 2, corresponding to quadratic cost in $\kappa$, while the slopes of all other lines are 3, 
    corresponding to cubic cost in $\kappa$. The line corresponding to the result by using Haah method to solve odd polynomials is not shown in the figure because only two data points are generated due to the numerical instability.}
    \label{fig:xcom}
\end{figure}

\begin{figure}[htbp]
    \centering
    \includegraphics[width=8cm]{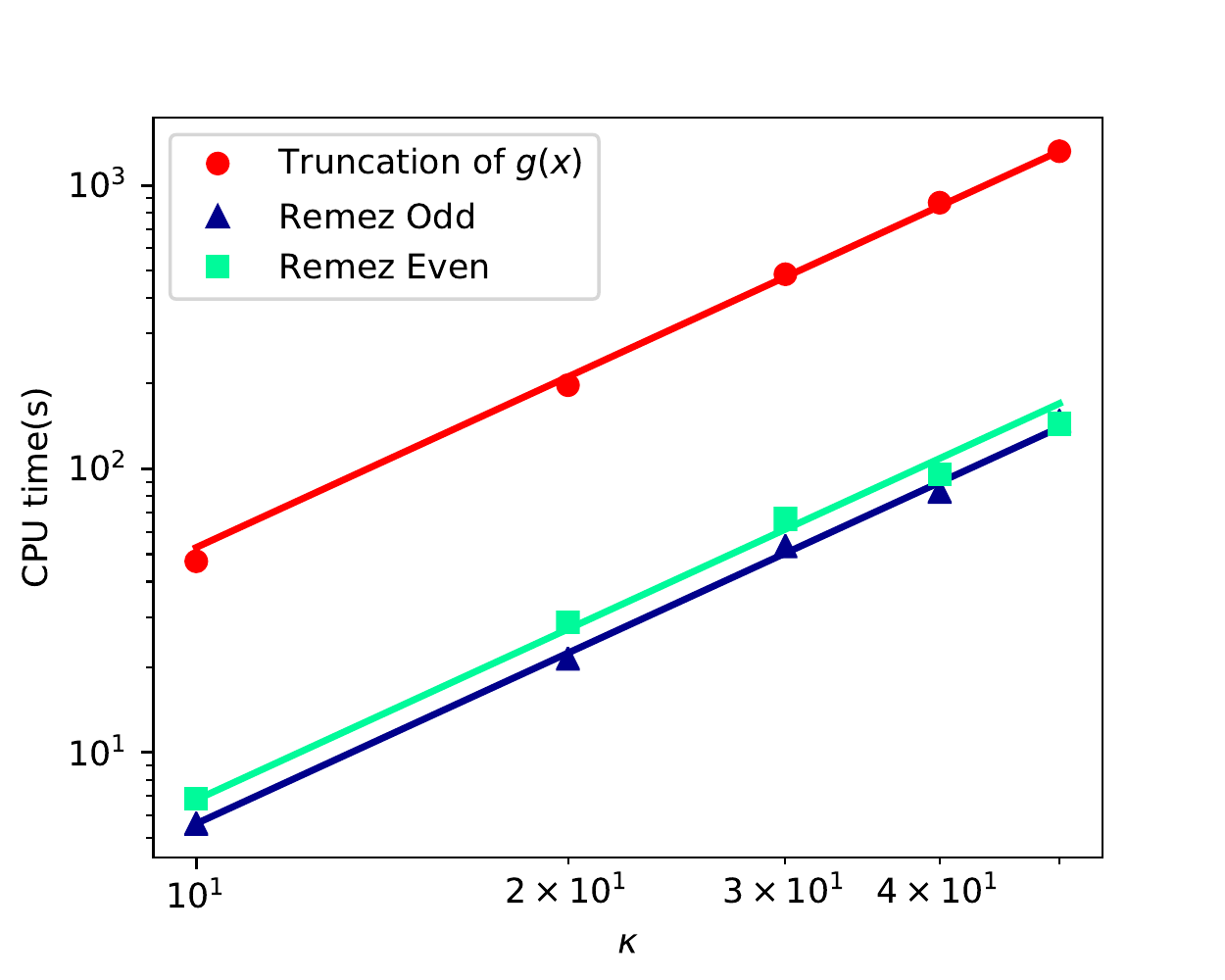}
    \caption{Comparison of CPU time(s) of optimization algorithm for approximating $1/x$ over $D_\kappa$ via QSP as a function of $\kappa$ for the two different methods of finding the optimal polynomial. The red (gray) line represents the result of approximating the polynomial with the Fourier method, \cref{eq:num-x}, the blue (dark gray) and the green (light gray) lines represent the CPU time of results of approximating the polynomial using the Remez method with odd and even parity, respectively. All lines have slope $2$, corresponding to quadratic cost in $\kappa$.}
    \label{fig:xinv}
\end{figure}

\subsection{Impact of the initial point}\label{app:initial}
To demonstrate the complexity of the optimization landscape, we report the final value of the objective function starting from randomly generated points for the Hamiltonian simulation problem. For $\tau\in\{100,200,300,400,500\}$, we choose the target polynomial
\begin{equation}
f(x)=J_0(\tau)/2+\sum_{\text{k even}}^\qspdeg (-1)^{k/2}J_k(\tau)T_k(x)
\end{equation}
as an approximation to $\cos(\tau x)/2$. The initial points are uniformly distributed in $[-\pi,\pi)^{\qspdeg+1}$. We run the L-BFGS algorithm until it converges or the number of iteration reaches 200. \cref{fig:initial} summarizes the performance of the algorithm under random initialization. We see that most of the calculations get stuck in local minima with a relatively large objective value, confirming the complexity of the landscape. Furthermore, the difficulty of finding a good solution increases with the degree of the polynomial. By comparison, if we start from 
$\Phi=(\frac{\pi}{4},0,\dots,0,\frac{\pi}{4})$, the algorithm will converge within dozens of iterations to the global minimum with the objective function very close to $0$. 

\begin{figure}[htbp]
    \centering
    \includegraphics[width=6.5cm]{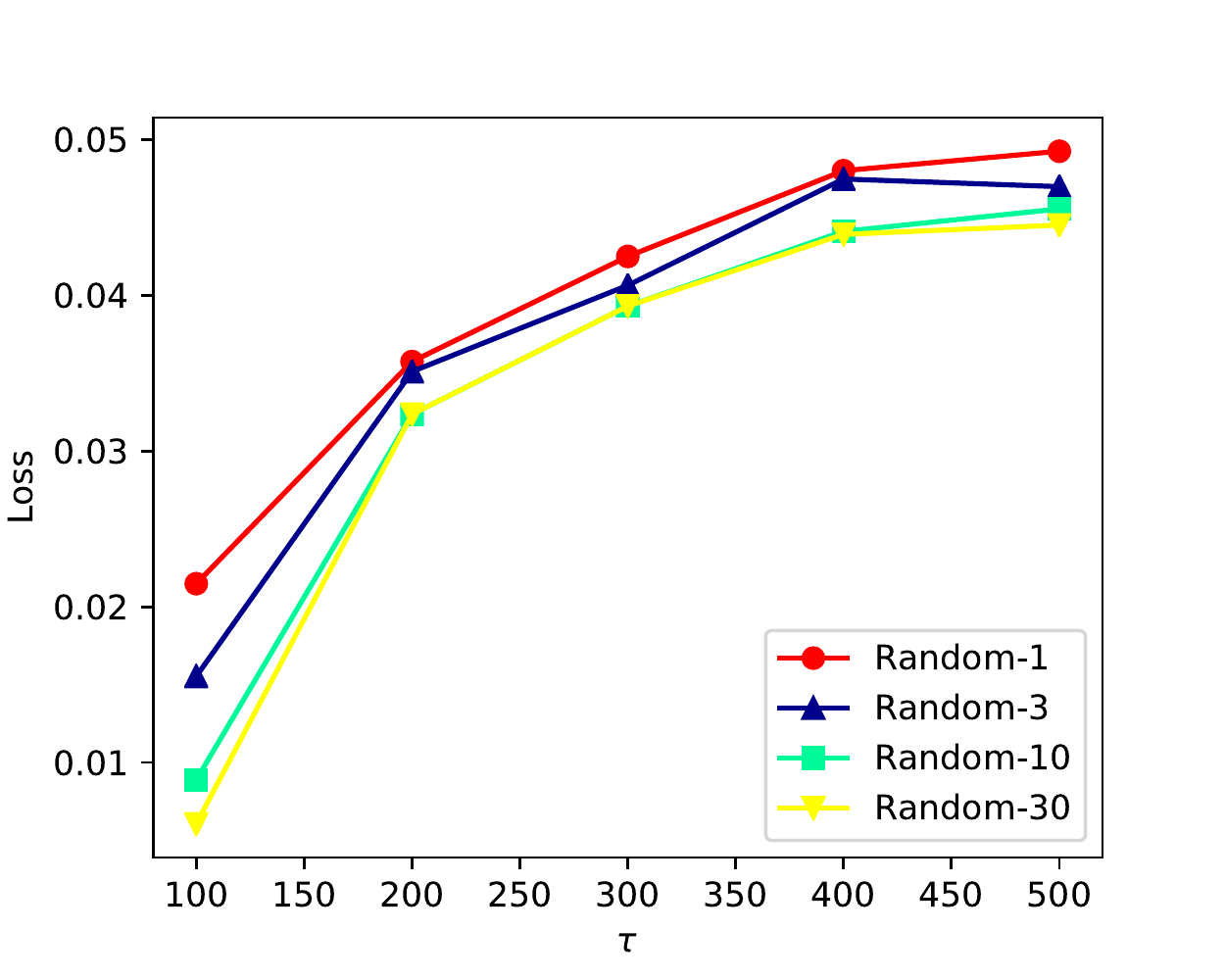}
    \caption{Loss of optimization method initiated with randomly generated points. The target polynomial is defined as the truncated polynomial of Jacobi-Anger expansion of degree $\qspdeg=1.4|\tau|+ \log( 1/\epsilon_0)$ with $\epsilon_0=10^{-14}$. ``Random-$k$'' represents that we start from $k$ different initial points and select best result.}
    \label{fig:initial}
\end{figure}

\subsection{Sensitivity analysis}\label{sec:sensitivity}

We further analyze the robustness of the method by reporting the condition number of the Hessian matrix $\text{Hess}\, L(\hat{\Phi}^*)$ at the optimal point. The condition number of the Hessian matrix is an indicator reflecting the sensitivity of the optimizer with respect to small perturbations of the target function.

We compute here the Hessian condition number for the three optimization problems presented above in Sections \ref{subsec:HamiltonianSimulation} - \ref{subsec:MatrixInverse}. Interestingly, we observe that the condition number is mostly affected by $L^\infty$ norm of the target polynomial, rather than by its degree or by its parameters. Thus, each problem can be exemplified by one polynomial with a given degree and parameters. To investigate how the norm affects Hessian condition number, we scale the $L^\infty$ norm of the given polynomial to $1-\eta$.  Fig. \ref{fig:condition} shows the scaled Hessian condition numbers as a function of $\eta$. As $\eta \rightarrow 0^+$, we find that the condition number increases as $\eta^{-\gamma}$ with $\gamma>1$ in all three cases.  
This indicates that when $\norm{f}_\infty$ is close to $1$, the optimizer  can be very sensitive to perturbations in $f$. When $\norm{f}_\infty$ is below $1$, the enhanced stability implies that these  phase factors can be used as an initial guess for a slightly perturbed target polynomial, which will be discussed in detail in \cref{sec:decay}. Furthermore, scaling the target polynomial $f$ to ensure that $\norm{f}_\infty \leq 1-\eta$ for some given threshold $\eta$ is also preferable. Such scaling of the target polynomial was also suggested in the root-finding procedures of the direct algorithms in order to ensure numerical stability \cite{Haah2019}.

\begin{figure}[htbp]
    \centering
    \includegraphics[width=6.5cm]{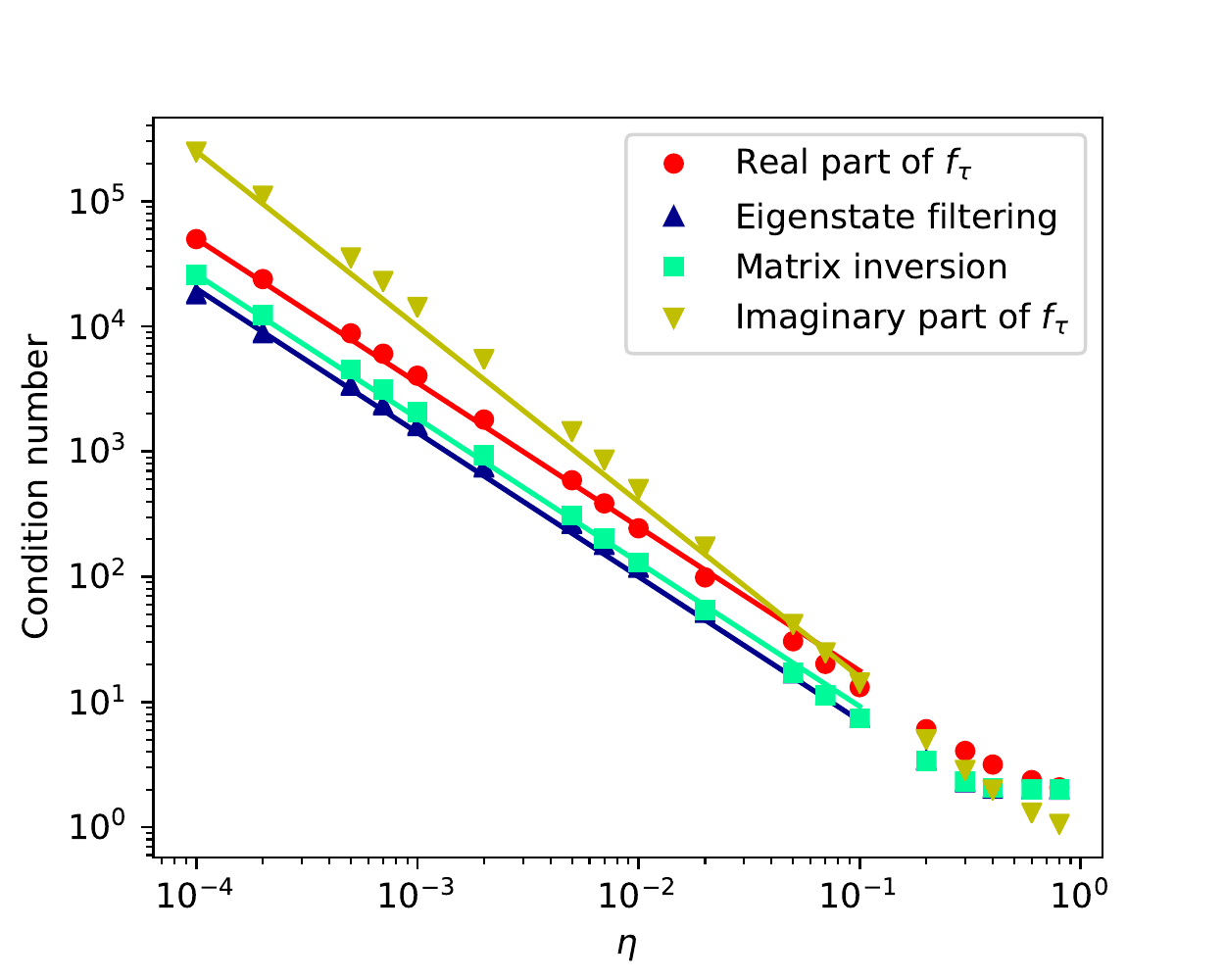}
    \caption{Condition number of the Hessian matrix at the optimum of the objective function $L(\tilde{\Phi})$ defined in \cref{eq:alg-prac}, shown for three different target polynomials studied in this work. 
    (a) Real (Imaginary) part of truncated Jacobi-Anger expansion in \cref{eq:num-jacobi}, where $\tau=200$ and $\qspdeg=312$ (represented by red dots and yellow downward triangles, respectively). (b) Eigenstate filter defined in \cref{eq:num-eigen} with $k=300$, $\Delta=0.05$ (represented by blue triangles). (c) Even polynomial approximation of $1/x$ on $D_{\kappa=20}$ generated by the Remez method (represented by green squares). Polynomials are scaled by a constant factor such that $\norm{f}_\infty=1-\eta$, where $\eta$ is the x-axis.  The slope is $1.4$ for the yellow (top) line and $1.15$ for all others.}
    \label{fig:condition}
\end{figure}

\section{Decay of phase factors from the center and phase factor padding}\label{sec:decay}

In addition to the symmetry structure discussed in \cref{sec:symmetry}, for smooth target functions, we observe that the QSP phase factors decay rapidly away from the center. To illustrate the decay and also the symmetry, we plot several examples in \cref{fig:qspdecay}. After subtracting the $\pi/4$ factor on both ends of the phase factors, we observe that the decay of the phase factors closely follows the decay of the Chebyshev coefficients (defined only on the positive axis in \cref{fig:qspdecay}). 

\begin{figure}[htbp]
    \centering
    \includegraphics[width=8.5cm]{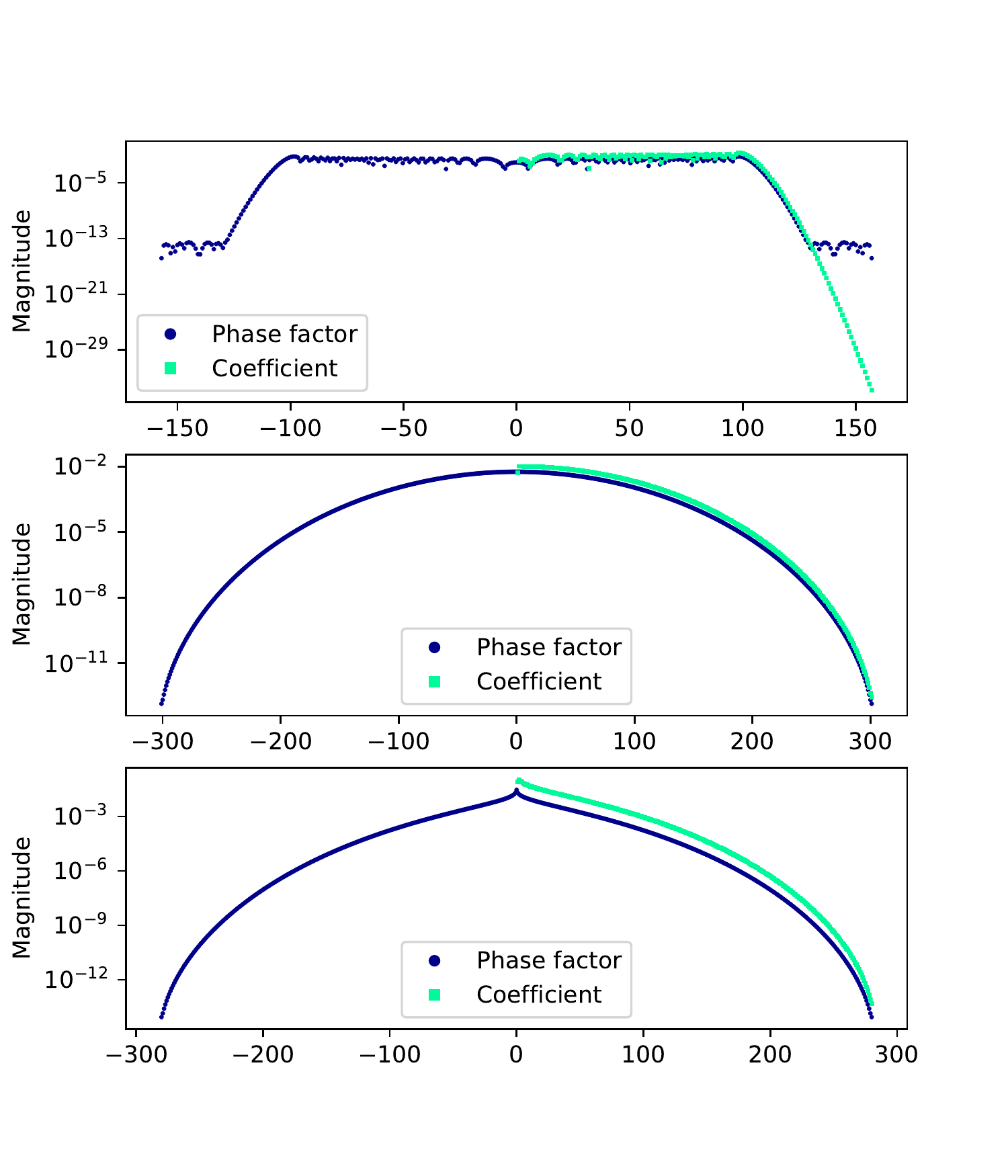}
    \caption{Magnitude of coefficients of the polynomial $f$ in the Chebyshev basis (light gray) and the corresponding phase factors (after subtracting $\pi/4$ on both ends, dark gray), for the three different problems studied in this work. We shift the x-axis to more clearly illustrate the symmetry property of the phase factors. Coefficients that are zero due to parity are omitted. (a) Real part of truncation of the Jacobi-Anger expansion in \cref{eq:num-jacobi}, with $\tau=200$ and $\qspdeg=312$. (b) Eigenstate filter defined in \cref{eq:num-eigen}, with $k=300$, $\Delta=0.05$. (c) Even polynomial approximation of $1/x$ on $D_{\kappa=20}$ as generated by the Remez method. }
    \label{fig:qspdecay}
\end{figure}

\cref{thm:qspcheby} states that for phase factors with relatively small magnitudes, the optimal phase factors can be expressed approximately analytically in terms of the coefficients of the Chebyshev polynomial expansion. The proof is given in \cref{sec:proofdecay}.

\begin{theorem}
    \label{thm:qspcheby}
    Let $\Phi \in \qspspace{\qspdeg^\prime}$  be a set of symmetric QSP phase factors. Define $\wt{\phi}_j := \phi_j - \frac{\pi}{4} \left(\delta_{j,0} + \delta_{j,\qspdeg^\prime-1}\right)$ and $\wt{\Phi} := (\wt{\phi}_0, \cdots, \wt{\phi}_{\qspdeg^\prime-1})$. Define a polynomial
    \begin{equation}
    \begin{split}
        &g_\Phi(x) := - \left(\prod_{j = 0}^{\qspdeg^\prime-1} \cos\wt{\phi}_j\right) \times \\
        & \left\{
        \begin{array}{ll}
            \displaystyle \sum_{j = 0}^\qspdeg 2\tan\left(\wt{\phi}_j\right) T_{2\qspdeg+1-2j}(x) &, \qspdeg^\prime = 2 \qspdeg + 2\\
            \displaystyle \tan\left( \wt{\phi}_{\qspdeg} \right) + \sum_{j = 0}^{\qspdeg-1} 2\tan\left(\wt{\phi}_j\right) T_{2\qspdeg-2j}(x) &, \qspdeg^\prime = 2 \qspdeg + 1.
        \end{array}
        \right.
    \end{split}
    \end{equation}
    Then for sufficiently small $\norm{\wt{\Phi}}_1$, there exists a constant $C>0$ such that the desired QSP component $f_\Phi(x) := \qspobj{\Phi}{x}$ satisfies
        \begin{equation}\label{eqn:decay_bound}
    \norm{f_\Phi(x) - g_\Phi(x)}_\infty \leq  C \norm{\wt{\Phi}}_1^3.
      \end{equation}
\end{theorem}

According to \cref{thm:qspcheby}, one can directly deduce approximate values of the phase factors from the coefficients of the Chebyshev expansion. For example, when $\qspdeg^\prime$ is even, $\wt{\phi}_j \approx - \arctan(c_{2\qspdeg+1-2j}/2) \approx  - c_{2\qspdeg+1-2j}/2$ holds up to $\Or(\Vert \wt{\Phi} \Vert_1^3)$. For smooth functions, the Chebyshev coefficients decay at least super-algebraically (\ie, faster than any polynomial decay) \cite{Boyd2001}. So the phase factors also decay super-algebraically away from the center.
The uniformly small phase factors can be realized by rescaling the function $f$ to $f/\beta$, with $\beta$ being a large number. We remark that our numerical results in \cref{fig:qspdecay} do not rely on such a scaling factor. A more precise characterization of the decay of the phase factors will be a focus of future work.

One possible usage of the decay property of the phase factors is as follows, which we refer to as a ``phase padding'' procedure.  Suppose we have solved the QSP phase factors corresponding to a polynomial approximation $f_1$ of relatively low degree to a real-valued function $f$ with definite parity. In order to improve the accuracy of the approximation, another small term $f_2$ of higher polynomial degree is needed to be added to approximate $f$ together with $f_1$. Therefore, a natural question is whether we can reuse the phase factors associated with $f_1$ to generate that of $f_1 + f_2 \approx f$. 

To solve this problem, one needs to increase the dimension of $\Phi$, since the degree of the polynomial has been increased and hence also the number of phase factors.  Due to the symmetry structure, we may consider the following symmetrically padded phase factors and further show that the symmetrical padding operation preserves the desired part of the QSP.

\begin{definition}[\textbf{$l$-padded phase factors}]
    Let $\Phi = (\phi_0, \cdots, \phi_\qspdeg) \in \qspspace{\qspdeg+1}$ be symmetric QSP phase factors. Then, the corresponding $l$-padded phase factors in $\qspspace{\qspdeg+2l+1}$ are given by $\Phi_l := (\frac{\pi}{4}, \underbrace{0, \cdots, 0}_{l-1}, \phi_0 - \frac{\pi}{4}, \phi_1, \cdots, \phi_{\qspdeg-1}, \phi_\qspdeg - \frac{\pi}{4}, \underbrace{0, \cdots, 0}_{l-1}, \frac{\pi}{4})$.
\end{definition}
\begin{theorem}
    \label{thm:lift}
    Given a set of symmetric phase factors $\Phi$ and a nonnegative integer $l$, its $l$-padded phase factors preserve the real part of the upper-left component of the QSP unitary matrix, \ie, $\qspobj{\Phi}{x} = \qspobj{\Phi_l}{x}, \forall x \in [-1, 1]$.
\end{theorem}

\begin{proof}
    Using \cref{lma:rot}, it is equivalent to prove the equality 
    $$\Im\left[ \langle 0 | U_\Phi(x) | 0 \rangle \right] = \Im\left[ \langle 0 | W(x)^l U_\Phi(x) W(x)^l | 0 \rangle \right]$$ 
for symmetric phase factors $\Phi$.  Insert the resolution of identity,
    \begin{equation*}
        \begin{split}
            \text{r.h.s.} =&  \Im \left[ T_l(x)^2 P(x) - 2(1-x^2)\seccheby_{l-1}(x)T_l(x)Q(x) \right.\\
            &\left. - (1-x^2) \seccheby_{l-1}(x)^2 P^*(x) \right]\\
             =& (T_l^2(x)+(1-x^2)\seccheby_{l-1}^2(x)) \Im[P(x)] \\
             =&  \Im[P(x)]= \text{l.h.s.}
        \end{split}
    \end{equation*}
    Here we have used  $Q \in \RR[x]$ according to \cref{thm:inv}.
\end{proof}

To demonstrate the usage of this phase  padding procedure, we consider the approximation of  $\cos(\tau x)/2$, namely, the real part of \cref{eq:num-jacobi} scaled by a constant factor $2$. First, an integer $\qspdeg_0$ is chosen such that the truncated series up to $\qspdeg_0$ is a rough approximation of $\cos(\tau x)/2$. Meanwhile, the corresponding phase factors are solved by optimization. Then we gradually increase the size of the problem by an even number $l$, \ie,  adding $l/2$ more terms of higher order polynomials. In order to reuse the phase factors, the initial guess in step $k$ is lifted from the phase factors solved in the previous step, \ie, the polynomial approximation of degree $\qspdeg_0+(k-1)l$. The procedure is repeated until the degree meets a maximal criterion $\qspdeg_1$, which generates an accurate polynomial approximation of $\cos(\tau x)/2$. 

The parameters in numerical implementations are set to be $\tau=500,\,\qspdeg_0=500,\,l=10,\,\qspdeg_1=600$. The $L^\infty$ error before the optimization (\ie, only using phase factor padding) and after the optimization in each step is shown in \cref{fig:padding}, while \cref{tab:padding} compares the computational cost between optimizations initiated with and without padding. We observe that the polynomial given by the lifted phase factors is already close to the target polynomial. This means that the lifted phase factors provide a good initial guess close to the global minimum. 

\begin{figure}[htbp]
    \centering
    \includegraphics[width=6.5cm]{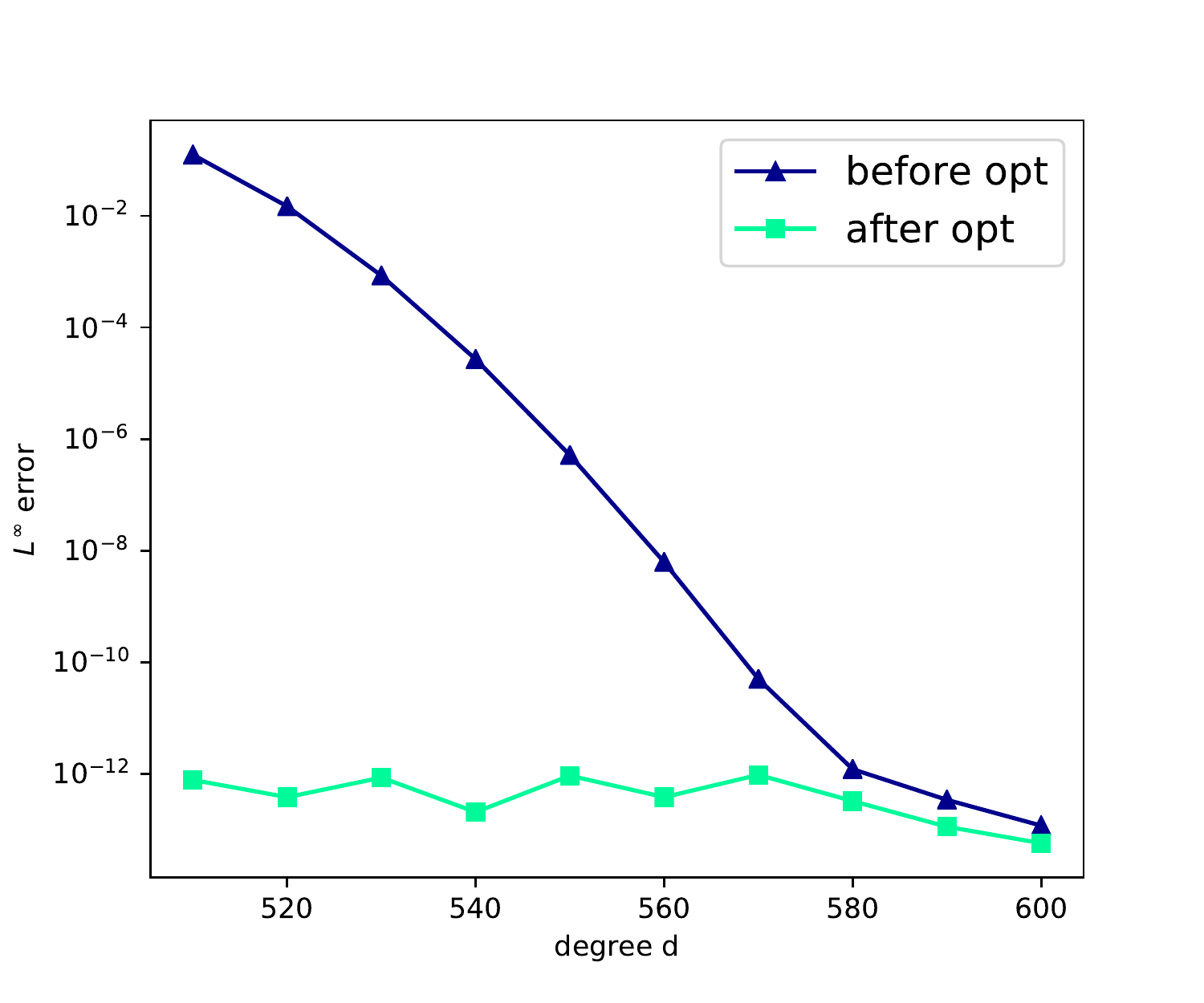}
    \caption{$L^\infty$ error between the polynomial obtained from the lifted phase factors, and the target polynomial, as a function of the degree $\qspdeg$ of the latter. Blue triangles represent the error before optimization, and green squares represent the error after optimization. The target polynomial here is the truncated polynomial of Jacobi-Anger expansion of degree $\qspdeg$. }
    \label{fig:padding}
\end{figure}

\begin{table}[htbp]
    \centering
    \begin{tabular}{|c|ccccc|}\hline
      $\qspdeg$ & 510 & 520 &530& 540& 550 \\ \hline
      with padding &19.9 & 19.7& 16.9&  16.4& 12.2  \\
      without padding & 21.8 & 21.2& 22.5& 22.6& 23.5  \\ \hline
     $\qspdeg$ & 560 & 570 &580& 590& 600 \\ \hline
    with padding& 9.37&4.69 & 3.18& 3.17 &3.19 \\ 
    without padding &24.2&26.1& 28.5& 28.3& 27.7  \\ \hline
    \end{tabular}
    \caption{CPU times for optimizations initiated with and without phase  padding (see text). The target polynomial here is the truncated polynomial of Jacobi-Anger expansion of degree $\qspdeg$. }
    \label{tab:padding}
\end{table}

\FloatBarrier

\section{Discussion}

We have demonstrated that using an optimization based approach, we can efficiently and accurately evaluate the phase factors needed to build QSP circuits for generation of unitary representations of non-unitary operations. Taken together with the QSP formalism of Refs. \cite{LowChuang2017,GilyenSuLowEtAl2018}, this approach now provides efficient and accurate constructive procedures to implement QSP and thereby removes a crucial bottleneck for the application of QSP in quantum algorithms.  We expect that our method 
will be useful for a wide range of matrix functions of interest to quantum algorithms, including the broad classes of Hamiltonian simulation, generation of thermal states, and linear algebra problems.  The optimization approach was found to be superior to previous direct methods that rely on a reduction procedure in which numerical errors are accumulated and amplified.  Instead of employing a reduction procedure, our approach is based on optimization of a distance function that quantifies the difference between the target polynomial and the QSP representation of this, with the QSP phases as variable parameters.  We identified two key features for success of the optimization based method: first, the choice of the initial guess, and second, preservation of the symmetry structure of the phase factors. We found that a simple choice of the initial guess can be surprisingly effective, despite the complexity of the global landscape of the objective function. This indicates that a better understanding of the local energy landscape connecting the initial guess to the optimal phase factors is needed.  Our study also reveals the connection between two seemingly unrelated objects in the QSP construction, namely, the decay of phase factors from the center, and the decay of the Chebyshev coefficients of the target function. More precise characterization of this connection will be a useful future research direction, together with further work to understand the energy landscape of the objective function.

\vspace{1em}
\textbf{Acknowledgment}
This work was partially supported by a Google Quantum Research Award (Y.D.,L.L.,B.W.), by the Quantum Algorithm Teams Program under Grant No. DE-AC02-05CH11231 (L.L. and B.W.), and by Department of Energy under Grant No. DE-SC0017867 (L.L.). X.M. thanks the Office of international relations, Peking University, Beijing, China for partial funding of an exchange studentship at the University of California, Berkeley. We thank Robert Kosut, Nathan Wiebe, and Yu Tong for discussion. Y.D. and X.M. contributed equally to this work.

\bibliographystyle{abbrvnat}

\widetext
\clearpage
\appendix

\section{Uniqueness of phase factors}\label{app:qsp_unique}
We refer to the representation in \cref{eq:qsp-gslw} as appeared in \cite[Theorem 3]{GilyenSuLowEtAl2019} as GSLW's representation. There is another equivalent form proposed in \cite{Haah2019}, which we call it Haah's representation. Under Haah's representation, the QSP unitary is 
\begin{equation}
\begin{split}
    U_{\hat{\Phi}}(x) &= e^{\I \sigma_z \hat{\phi}_0} \prod_{j = 1}^\qspdeg \left[ e^{\I \sigma_z \hat{\phi}_j / 2} W(x) e^{- \I \sigma_z \hat{\phi}_j / 2} \right] = e^{\I \sigma_z (\hat{\phi}_0 + \hat{\phi}_1/2)} \left(\prod_{j = 1}^{\qspdeg-1} W(x) e^{\I \sigma_z (\hat{\phi}_{j+1} - \hat{\phi}_j)/2}\right) W(x) e^{- \I \sigma_z \hat{\phi}_\qspdeg / 2}
\end{split}
\end{equation}
where $\hat{\phi}_{\qspdeg+1} := 0$. Compared to \cref{eq:qsp-gslw}, the transformation between two representations is evident, \ie, $\mathcal{T} : [-\pi, \pi)^{\qspdeg+1} \rightarrow \qspspace{\qspdeg+1},\ \hat{\Phi} \mapsto \Phi$ such that $\phi_0 = \hat{\phi}_0 + \frac{\hat{\phi}_1}{2},\ \phi_j = \frac{\hat{\phi}_{j+1} - \hat{\phi}_j}{2},\ \forall j = 1, \cdots, \qspdeg-1$ and $\phi_\qspdeg=-\hat{\phi}_\qspdeg / 2$.   The irreducible set $\qspspace{\qspdeg+1}$ is defined as the image of this linear transformation. The uniqueness of Haah's phase factors in $[-\pi, \pi)^{\qspdeg+1}$ was proved in \cite[Theorem 2]{Haah2019}, which considers a formally more general class of polynomial functions $\text{U}(1) \rightarrow \text{SU}(2)$. The bijection $\mathcal{T}$ implies the uniqueness of GSLW's phase factors in $\qspspace{\qspdeg+1}$. It is evident that the $2\pi$-periodicity of Haah phase factors lead to a pair of $\pm \pi$ shifts in the corresponding GSLW phase factors. Then, if we define the equivalence relation $\Phi \sim \Psi$ when $\phi_k = \psi_k, \forall k \neq i, j$ and $\phi_i = \psi_i + \pi, \phi_j = \psi_j - \pi$, the irreducible set is the quotient space $\qspspace{\qspdeg+1} \equiv [-\pi, \pi)^{\qspdeg+1} / \sim$.

\section{Reducing one ancilla qubit for representation of real polynomials}\label{app:qsp_real_save}

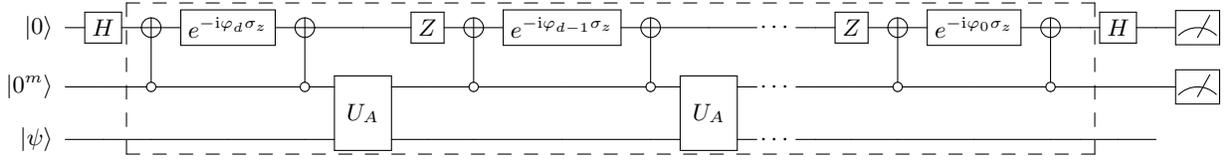
\begin{figure*}[htb]
\begin{center}
\[
\Qcircuit @C=0.8em @R=1.em {
 \lstick{\ket{0}}& \gate{H} & \targ & \gate{e^{-\I \varphi_\qspdeg \sigma_z}} & \targ & \qw &\gate{Z}& \targ & \gate{e^{-\I \varphi_{\qspdeg-1} \sigma_z}} & \targ & \qw & \qw &\raisebox{0em}{$\cdots$}&&\qw   &\gate{Z}&\targ & \gate{e^{-\I \varphi_0 \sigma_z}} & \targ & \qw& \gate{H}&\qw &\meter   \\
\lstick{\ket{0^m}}& \qw  &\ctrlo{-1} & \qw  & \ctrlo{-1} & \multigate{1}{U_A} & \qw& \ctrlo{-1} & \qw & \ctrlo{-1} & \multigate{1}{U_A} &\qw &\raisebox{0em}{$\cdots$} &&\qw & \qw   &\ctrlo{-1} & \qw & \ctrlo{-1} & \qw& \qw &\qw &\meter\\
\lstick{\ket{\psi}}& \qw &\qw &\qw&\qw &\ghost{U_A} &\qw&\qw&\qw&\qw& \ghost{U_A}&\qw &\raisebox{0em}{$\cdots$} &&\qw&\qw&\qw&\qw   & \qw&\qw& \qw &\qw 
\gategroup{1}{3}{3}{20}{1.2em}{--}
}
\]
\end{center}
\caption{Quantum circuit for quantum signal processing of real matrix polynomials with a Hermitian block-encoding $U_A$. }
\label{fig:qsp_circuit_mod}
\end{figure*}

We explain here why the additional ancilla qubit needed for representing real polynomials  in \cref{sec:matrixpolynomial}, case 1 as a result of the linear combination of two QSP circuits, is in fact not needed and can be avoided. Specifically, this ancilla qubit can  be  combined with the first ancilla qubit in \cref{fig:qsp_circuit}. To see why this is the case, note that the phase factors for $U_{-\Phi}$ in \cref{eqn:neg_phi_varphi_relation} can be obtained by taking the phase factors for $U_{\Phi}$ in \cref{eqn:phi_varphi_relation}, and perform the mapping $\varphi_i\mapsto -\varphi_i+\pi(1-\delta_{i\qspdeg}),i=0,\ldots,d$. In other words, we negate $\varphi_i$ and add $\pi$ to all but the $\qspdeg$-th entry. Negating the phase can be implemented by feeding $\ket{1}$ instead of $\ket{0}$ to the signal state, and adding $\pi$ to the phase can be implemented via a $\sigma_z$ gate associated with $\phi_0,\ldots,\phi_{\qspdeg-1}$. 

We may verify that by slightly modifying \cref{fig:qsp_circuit}, the circuit in the box with dashed line in \cref{fig:qsp_circuit_mod} in fact implements
\[
\ket{0}\bra{0} \otimes U_{\Phi} + \ket{1}\bra{1} \otimes U_{-\Phi}, 
\]
which is the select oracle, \cref{eqn:select_oracle}. Therefore, using the Hadamard gate as the prepare oracle (\cref{eqn:prepare_oracle}) as before, the circuit \cref{fig:qsp_circuit_mod} provides a $(1,m+1,0)$-block-encoding of $f(A/\alpha)$, which saves one ancilla qubit.

\section{Quantum signal processing with a non-Hermitian block-encoding matrix}\label{app:qsp_general_block_encode}

Let $A$ be an $n$-qubit Hermitian matrix, but its $(\alpha,m,0)$-block-encoding $U_A$ is not Hermitian. We can still perform QSP by introducing an additional ancilla qubit. To this end, we first generate an $(\alpha,m+1,0)$-block-encoding  of $A$ that is Hermitian. Define an $(m+n+1)$-qubit controlled block-encoding as
\begin{equation}
V^\prime_A := |0\rangle\langle 0|\otimes U_A+| 1\rangle\langle 1| \otimes U_A^{\dagger},
\label{eqn:general_control_UA}
\end{equation}
which uses both $U_A$ and $U_A^{\dag}$. 
We also introduce  the swap operation $S := \sigma_x \otimes I_m$. Then
\begin{equation}
U_A^{\prime}:=(S\otimes I_n)V_A^{\prime}=|1\rangle\langle 0|\otimes U_A+| 0\rangle\langle 1| \otimes U_A^{\dagger}
\label{eqn:}
\end{equation}
is Hermitian. Define an $(m+1)$-qubit signal state for block-encoding
 \begin{equation}
|G\rangle := \ket{+}\ket{0^m}=\frac{1}{\sqrt{2}} (|0\rangle + |1\rangle)|0^m\rangle,
\label{eqn:general_signal}
\end{equation}
then
\begin{displaymath}
\begin{split}
(\bra{G}\otimes I_n)U_A^{\prime}(\ket{G}\otimes I_n)=&(\bra{G}\otimes I_n)V_A^{\prime}(\ket{G}\otimes I_n)\\
=&\frac12 (\bra{0^m}\otimes I_n)U_A(\ket{0^m}\otimes I_n)+\frac12 (\bra{0^m}\otimes I_n)U_A^{\dag}(\ket{0^m}\otimes I_n)\\
=& \frac12 A+ \frac12 A^{\dag}=A.
\end{split}
\end{displaymath}
In the last equality, we used that $A$ is a Hermitian matrix. This proves that 
$U_A^{\prime}$ is indeed an $(\alpha,m+1,0)$-block-encoding  of $A$.  Define 
\begin{equation}
U^{\prime}_{\Pi}=(2|G\rangle\langle G|-I_{m+1}) \otimes I_{n},
\label{eqn:general_control_UPi}
\end{equation}
we may use Jordan's lemma to simultaneously block-diagonalize $U^{\prime}_{\Pi},U^{\prime}_{A}$. In particular, the matrix representation in \cref{eqn:twoblock_representation} still holds, which provides the qubitization of  $A$. 

Then QSP representation in \cref{eqn:QSP_representation_1,eqn:QSP_representation_2} can be directly obtained by substituting $U_{\Pi}\to U^{\prime}_{\Pi}, U_{A}\to U^{\prime}_{A}$. The circuit is given in \cref{fig:general_qsp_circuit}.
In the second line, the Hadamard gate converts the $\ket{+}$ state in the signal state into $\ket{0}$ and back in order to apply the $(m+2)$-qubit Toffoli gate. The swap operation can be implemented via a single $\sigma_x$ gate. The last Hadamard gate in the second line is not present, in order to measure in the $\ket{\pm}$ basis set according to the signal state $\ket{G}$.

\begin{figure*}[htb]
\begin{center}
\[
\Qcircuit @C=0.7em @R=1.0em {
 \lstick{\ket{0}}& \qw & \targ & \gate{e^{-\I \varphi_\qspdeg \sigma_z}} & \targ & \qw & \qw & \qw & \qw & \qw & \raisebox{0em}{$\cdots$}&&\qw &\qw & \targ & \gate{e^{-\I \varphi_0 \sigma_z}} & \targ & \qw & \qw &\qw \\
\lstick{\ket{+}}& \gate{H} & \ctrlo{-1} & \qw & \ctrlo{-1} & \gate{H} & \ctrlo{1} & \ctrl{1} & \gate{\sigma_x} & \qw & \raisebox{0em}{$\cdots$}&& \qw &\gate{H} & \ctrlo{-1} & \qw & \ctrlo{-1} & \qw & \qw & \meter\\
\lstick{\ket{0^m}}& \qw & \ctrlo{-1} & \qw & \ctrlo{-1} & \qw & \multigate{1}{U_A} & \multigate{1}{U_A^{\dag}}  & \qw &\qw & \raisebox{0em}{$\cdots$}&&\qw & \qw & \ctrlo{-1} & \qw & \ctrlo{-1} & \qw & \qw &\meter\\
\lstick{\ket{\psi}}&\qw&\qw&\qw&\qw&\qw&\ghost{U_A} &\ghost{U_A^{\dag}}& \qw&\qw&\raisebox{0em}{$\cdots$} &&\qw&\qw&\qw& \qw&\qw&\qw&\qw&\qw
\gategroup{1}{2}{4}{9}{2.2em}{--}
}
\]
\end{center}
\caption{Quantum circuit for quantum signal processing with a non-Hermitian block-encoding matrix. The circuit in the box enclosed by the dashed line should be repeated $\qspdeg$ times, each time with a different phase factor. The last Hadamard gate in the second line is removed if measurements are to be made in the $\ket{\pm}$ basis set. }
\label{fig:general_qsp_circuit}
\end{figure*}
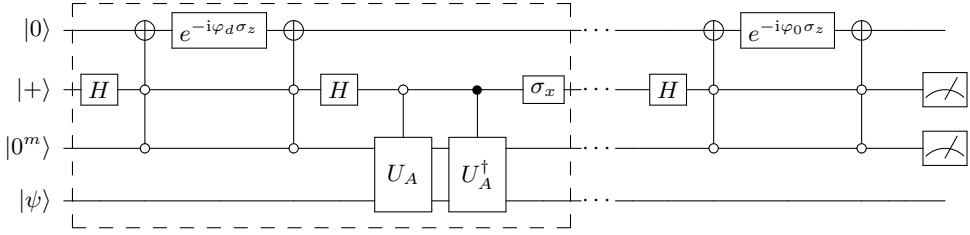

\section{Proof of \cref{thm:optmodel}}
\label{app:proof_cheby}

We first review some basic facts of the Chebyshev polynomial. The Chebyshev polynomials are two sequences of polynomials which can be defined by trigonometric functions. For each $\qspdeg \in \NN$ and $x \in [-1, 1]$, the Chebyshev polynomial of the first kind is defined as $T_\qspdeg(x) = \cos(\qspdeg\arccos(x))$ and that of the second kind is $\seccheby_\qspdeg(x) = \sin((\qspdeg+1)\arccos(x)) / \sin(\arccos(x))$. Both $T_\qspdeg$ and $\seccheby_\qspdeg$ are polynomials of degree $\qspdeg$. We will focus on the properties of Chebyshev polynomials of the first kind in the following context and call $T_\qspdeg$'s Chebyshev polynomials for simplicity. Define the weighted inner product as $(f, g)_w := \int_{-1}^1 f(x) g(x) \frac{\ud x}{\sqrt{1-x^2}}$ on the space $L^2_w([-1, 1])$. Then Chebyshev polynomials are orthogonal polynomials on $[-1, 1]$ with respect to the inner product $(\cdot, \cdot)_w$, and  form a complete basis on the space $L^2_w([-1, 1])$.
\begin{lemma}
    Any function $g \in L^2_w([-1, 1])$ can be uniquely expressed as a series of Chebyshev polynomials,
    $$g(x) = \sum_{j \in \NN} c_j T_j(x), \quad \text{where} \ c_j = \frac{2 - \delta_{j0}}{\pi} (g, T_j)_w.$$
\end{lemma}
By substituting $x \rightarrow \cos \theta$, the series in terms of Chebyshev polynomial becomes the Fourier series of periodic function $g(\cos\theta)$.
The roots of Chebyshev polynomials are called Chebyshev nodes, \eg, $\left\{ \cos\left( \frac{2 j - 1}{2 \qspdeg} \pi \right) : j = 1, \cdots, \qspdeg \right\}$ are Chebyshev nodes of $T_\qspdeg$. Chebyshev polynomials satisfy the discrete orthogonality
\begin{equation}
\sum_{j = 1}^\qspdeg T_m(x_j) T_n(x_j) = \qspdeg \frac{1 + \delta_{m,0}}{2} \delta_{m,n}
\label{eqn:discrete_chebyshev}
\end{equation}
where $\qspdeg>\lfloor (m+n) / 2 \rfloor$ is an integer and $x_j$'s are Chebyshev nodes of $T_\qspdeg$.

\begin{proof}(\cref{thm:optmodel})
Let $\tilde \qspdeg=\lceil\frac{\qspdeg+1}{2}\rceil$, then $2\tilde \qspdeg > \qspdeg$. Apply the Cauchy--Schwarz inequality, we have
\begin{equation}
\sum_{j=1}^{\tilde \qspdeg} |f(x_j)-f_\Phi(x_j)| \le \sqrt{\tilde \qspdeg\sum_{j=1}^{\tilde \qspdeg}|f(x_j)-f_\Phi(x_j)|^2 } = \sqrt{\tilde \qspdeg^2 L(\phi)} \le \tilde \qspdeg \sqrt \epsilon,   
\end{equation}
where $x_j=\cos\left(\frac{(2j-1)\pi}{4\tilde \qspdeg}\right),\,j=1,\dots,{\tilde \qspdeg}$ are positive roots of Chebyshev polynomial $T_{2\tilde \qspdeg}(x)$. For a fixed integer $t\leq \qspdeg$,
\begin{equation}
\sum_{j=1}^{\tilde \qspdeg} (f(x_j)-f_\Phi(x_j))T_t(x_j)=\sum_{j=1}^{\tilde \qspdeg} \sum_{m=0}^\qspdeg (\alpha_m-\beta_m)T_m(x_j)T_t(x_j)=\sum_{m=0}^{\qspdeg}(\alpha_m-\beta_m)\sum_{j=1}^{\tilde \qspdeg} T_m(x_j)T_t(x_j)=\sum_{m=0}^\qspdeg (\alpha_m-\beta_m)\eta_{mt},
\end{equation}
where by discrete orthogonality in \cref{eqn:discrete_chebyshev} and symmetry, $\eta_{mt} = \tilde \qspdeg \frac{1 + \delta_{m,0}}{2} \delta_{m, t}$. Thus we have 
\begin{equation}
|\alpha_m-\beta_m|\le \frac{2}{\tilde \qspdeg} \sum_{j=1}^{\tilde \qspdeg} |(f(x_j)-f_\Phi(x_j))T_m(x_j)|\le 2\sqrt \epsilon
\end{equation}
for any $m=0,\dots,\qspdeg$.
\end{proof}

\section{Remez Method}\label{sec:remez}
We would like to solve for the best approximation polynomial in terms of the $L^\infty$ norm
\begin{equation}
    f^*=\argmin{f\in \mathbb{R}[x], \deg(f)\le \qspdeg}{\max_{x\in[a,b]}|F(x)-f(x)|}.
\end{equation}
In addition, the approximation problem encountered in this work requires that the approximation polynomial has a definite parity. Hence, we need to focus on the best approximation problem, over the linear combination of a general basis of functions $\{g_1(x),\dots,g_N(x)\}$ other than $\{1,x,\dots,x^d\}$. In this paper we choose $N=\lceil\frac{d+1}{2}\rceil$, where $d$ is the degree of the approximation polynomial we would like to generate.  A series of functions $\{g_1(x),\dots,g_N(x)\}$ is said to satisfy the Haar condition on a set $X$, if each $g_j(x)$ is continuous and for every $N$ points $x_1,\dots,x_N\in X$, the $N$ vectors $v_j:=(g_1(x_j),\dots,g_N(x_j)),\,1\le j\le N$ are linearly independent \cite{haar1917minkowskische}.  As an example, the Haar condition holds if we choose $g_j(x)=T_{2j-1}(x)$ (or $T_{2j-2}(x)$) and $X\subset (0,1]$. Solution of the best approximation problem over such a basis will yield the best odd (even) approximation polynomial. Imposing the Haar condition simplifies the solution of the generalized approximation problem.

The optimal approximate polynomial $f^*$ over the linear combination of functions $\{g_1(x),\dots,g_N(x)\}$ can be found via the Remez exchange method summarized in \cref{alg:remez}, which computes a series of approximation polynomials on discrete sets. The polynomials $f_t$ generated by the Remez algorithm converge uniformly to the optimal polynomial $f^*$ with linear convergence rate. For a large range of functions $F$, the convergence rate can be improved to be quadratic. We refer the reader to \cite[Chapter~3]{cheney1966introduction} for more details related to the Remez method.

\begin{algorithm}[htbp] 
\caption{Remez method for solving the best approximation polynomial}
\label{alg:remez}
\begin{algorithmic} 
\STATE \textbf{Input:} An interval $[a,b]\subset \mathbb{R}$, target function $F$, a basis $\{g_1,\dots,g_N\}$ satisfying the Haar condition, $N+1$ initial points $a\le x_0\le\dots\le x_{N}\le b$.
\vspace{1em}
\STATE Set $t=0$.
\WHILE{stopping criterion is not satisfied}
\STATE Set $t=t+1$.
\STATE Solve the linear equation for $a_1,\dots,a_N$ and $\Delta$
$$
\sum_{j=1}^N a_jg_j(x_k) -F(x_k) =(-1)^k \Delta, \,\,k=0,\dots,N.
$$
\STATE Denote $f_t(x)=\sum_{j=1}^N a_jg_j(x)$ and residual $r(x)=F(x)-f_t(x)$. 
\STATE $r(x)$ has a root $z_j\in(x_{j-1},x_j)$ for $j=1,\dots,N$. Set $z_0=a$ and $z_{N+1}=b$.

\STATE Let $\sigma_j=\mathrm{sgn}(r(x_j))$. Find $y_j=\argmax{y\in[z_j,z_{j+1}]}\sigma_j r(y)$ for each $j=0,\dots,N$. 
\IF{$\|r(x)\|_\infty>\max_j|r(y_j)|$}
\STATE Choose $$y=\argmax{y\in[a,b]} |r(y)|.$$
\STATE Replace a $y_k\in\{y_0,\dots,y_N\}$ by $y$ in such a way that the values of $r(y)$ on the resulting ordered set still satisfies $$r(y_j)r(y_{j+1})<0,\,j=0,\dots,N-1.$$
\ENDIF
\STATE Replace $\{x_0,\dots,x_N\}$ by $\{y_0,\dots,y_N\}$. 

\ENDWHILE
\STATE \textbf{Output:} an approximation to the best approximation polynomial $f_t(x)$
\end{algorithmic}
\end{algorithm}

\section{L-BFGS Algorithm}\label{sec:l-bfgs}
In numerical optimization, the Broyden--Fletcher--Goldfarb--Shanno (BFGS) algorithm is a quasi-Newton method for solving unconstrained optimization problems \cite[Chapter~5]{sun2006optimization}.  The BFGS method stores a dense $n\times n$ matrix to approximate the inverse of Hessian matrix. It updates this approximation by performing a rank two update using gradient information along its trajectory. Limited-memory BFGS (L-BFGS) approximates the BFGS method by using a limited amount of computer memory \cite[Chapter~5]{sun2006optimization}. In particular, it represents the inverse of Hessian matrix implicitly by only a few vectors. For completeness, we summarize the procedure for the L-BFGS method in \cref{alg:lmbfgs}.

\begin{algorithm}[htbp]
\caption{\textbf{Function:} $\phi=\text{L-BFGS}(\phi^0,L,m,B^0)$} 
\label{alg:lmbfgs}
\begin{algorithmic} 
\STATE {\textbf{Input:} Initial point $\phi^0$, objective function $L(\phi)$, a nonnegative integer $m$ and initial approximation of  inverse Hessian $B^0$.\setlength{\belowdisplayskip}{0pt} \setlength{\belowdisplayshortskip}{0pt}\setlength{\abovedisplayskip}{0pt} \setlength{\abovedisplayshortskip}{0pt}}
\vspace{1em}
\STATE Set $t=0$
\WHILE{stopping criteria does not meet}
    \STATE  Compute $g_t=\nabla L(\phi^t)$, set $q=g_t$ 
    \FOR{$i=t-1,\dots,t-m$}
        \STATE Set $\alpha_i=\rho_i s_i^{\top}q$ 
        \STATE $q=q-\alpha_iy_i$
    \ENDFOR
    \STATE  $r=B_0q$ 
    \FOR{$i=t-m,\dots,t-1$}
        \STATE $\beta=\rho_i y_i^{\top}r$ 
        \STATE $r=r+s_i(\alpha_i-\beta)$ 
    \ENDFOR
    \STATE Set search direction $d_t=-r$. 
    \STATE Find a step size $\gamma_t$ using backtracking line search. 
    \STATE {Set \setlength{\belowdisplayskip}{0pt} \setlength{\belowdisplayshortskip}{0pt}
    \setlength{\abovedisplayskip}{0pt}
    \begin{align*}
        \phi^{t+1}=\phi^t+\gamma_td_t,&s_k=\phi^{t+1}-\phi^t,\\ 
        y_t=g_{t+1}-g_t,&\rho_t=\frac{1}{s_t^{\top}y_t}.
    \end{align*}}
    \STATE Set $t=t+1$.
\ENDWHILE
\STATE \textbf{Return:} $\phi^t$
\end{algorithmic}
\end{algorithm}

\section{Implementation details of the direct methods for finding phase factors}\label{app:implement}
For completeness we provide here our implementation of the direct methods for computing phase factors, \ie, the GSLW method and the Haah method. The codes are written in \textsf{Julia} v1.2.0. Although advanced root-finding algorithm with guaranteed performance \cite{pan1996optimal} is suggested in the Haah method \cite{Haah2019}, this is a theoretical result and hard to implement. We utilize instead the function \textsf{roots} in the \textsf{PolynomialRoots} package in \textsf{Julia} to find the roots of polynomials. For both GSLW and Haah methods, we perform calculations with variable precision arithmetic (VPA) using the \textsf{BigFloat} data type. The numbers of bits $R$ used in our numerical tests are empirical  parameters whose values are chosen to minimize CPU time while maintaining accuracy. We first take $R$ to be a large number and then gradually decrease it, until the algorithm fails to yield phase factors with sufficient accuracy. The algorithm is considered as a failure on an example if it cannot generate accurate enough phase factors, \ie, within the specified tolerance, despite the arithmetic being performed under increasingly high precision. Specifically, we choose $R=3\qspdeg$ for the GSLW method and $R=4\qspdeg$ for the Haah method in the Hamiltonian simulation, $R=2\qspdeg$ for both methods for the eigenstate filtering function, and $R=50\kappa$ for both methods in the matrix inversion problem. Here $\qspdeg$ is the degree of the polynomial.  Note that the polynomials encountered in the matrix inversion subsection approximate $1/x$ on $D_\kappa=[1/\kappa,1]$.

Our implementation of the  GSLW algorithm proposed in \cite{GilyenSuLowEtAl2019} is summarized in \cref{alg:GSLW}. To avoid stability issues caused by inaccurate roots, a root $s$ is regarded as a real (pure imaginary) number if the magnitude of its imaginary (real) part is smaller than machine precision ($\epsilon=10^{-16}$ in our implementation). Similarly, $s$ is rounded to $1$ if $|1-s|<\epsilon$. We evaluate the coefficients of $B(x)$ and $C(x)$ with respect to the Chebyshev basis by  discrete fast Fourier transform (FFT) to enhance numerical stability. The reduction procedure in the loop  is also performed based in the Chebyshev basis. We observe that compared to the original implementation of the GSLW method in \cite{GilyenSuLowEtAl2018}, the use of the Chebyshev basis significantly improves the numerical stability of the algorithm. Since in the examples in this work we primarily consider situations where only $P$ is required, we employ a zero polynomial as the input for the second polynomial $\tilde Q$.
  
\begin{algorithm}[htbp] 
\caption{GSLW method}
\label{alg:GSLW}
\begin{algorithmic} 
\STATE \textbf{Input:} A nonnegative integer $\qspdeg$, real polynomials $\tilde P$ and $\tilde Q$ satisfying condition (1) -- (2) of \cref{thm:qsp} and $\tilde P^2(x)+(1-x^2)\tilde Q^2(x)\le 1,\,\forall x\in[-1,1]$. A nonnegative integer $R$ indicates the number of bits on which high-precision arithmetic is performed.
\vspace{1em}
\STATE \textbf{Step 1:} Find the complementary polynomials
\vspace{1em}
\STATE Solve all roots of $1-P^2(x)-(1-x^2)Q^2(x)$. Denote $S$ as the multiset that contains roots of $1-P^2(x)-(1-x^2)Q^2(x)$ with their  algebraic multiplicity. Find the following subsets of S
\begin{align*}
&S_0=\{s\in S|s=0\}, & S_{(0,1)}=\{s\in S|s\in (0,1)\}, \\& S_{[1,\infty)}=\{s\in S|s\in [1,\infty)\}, &
S_I=\{s\in S|\Re(s)=0,\Im(s)>0\},\\ &S_C=\{s\in S|\Re(s)>0,\Im(s)>0\}.&
\end{align*}
\STATE Define
\begin{equation}
\begin{aligned}
Z(x)&=Kx^{|S_0|/2}\prod_{s\in S_{(0,1)}} \sqrt{x^2-s^2} \prod_{s\in S_{[1,\infty)}}(\sqrt{s^2-1}x+\I s\sqrt{1-x^2})\\
&\prod_{s\in S_I}(\sqrt{|s|^2+1}x+\I |s|\sqrt{1-x^2})\prod_{(a+b\I)\in S_C}(cx^2-(a^2+b^2)+\I \sqrt{c^2-1}x\sqrt{1-x^2}),
\end{aligned}
\end{equation}
where $K$ is the absolute value of the coefficient of the highest order of polynomial $1-P^2(x)-(1-x^2)Q^2(x)$, $c=a^2+b^2+\sqrt{2(a^2+1)b^2+(a^2-1)^2+b^4}$.

\STATE $Z(x)$ can be written in the form $Z(x)=B(x)+\I \sqrt{1-x^2}C(x)$ for $B,\,C\in \mathbb{R}[x]$. $B$ and $C$ are required complementing polynomials if $B$ has same parity as $\tilde P$ while $C$ has opposite parity, otherwise we replace $Z(x)$ by $Z(x)(x+\I \sqrt{1-x^2})$.

\STATE Calculate coefficients of $B$ and $C$ and define $P(x):=\tilde P(x)+\I B(x)$, $Q(x):=\tilde Q(x)+\I C(x)$. Then $|P(x)|^2+(1-x^2)|Q(x)|^2=1,\forall x\in[-1,1]$.

\vspace{1em}
\STATE \textbf{Step 2:} Matrix reduction
\vspace{1em}

\STATE Set $t=\qspdeg$.
\WHILE{$\deg(P)>0$}
\STATE Denote coefficients of highest order of $P$ and $Q$ as $p_t$ and $q_{t-1}$, respectively. We have $|p_t|=|q_{t-1}|$. Choose $\phi_t\in \mathbb{R}$ such that $e^{2\I \phi_t}=p_t/q_{t-1}$.

\STATE Replace $P$ and $Q$ by 
\begin{equation}
P_{\text{new}}(x) = e^{-\I \phi_t} \left(xP(x)+\frac{p_t}{q_{t-1}}(1-x^2)Q(x) \right)
\end{equation}
and
\begin{equation}
Q_{\text{new}}(x) = e^{-\I \phi_t} \left(\frac{p_t}{q_{t-1}}xQ(x)-P(x)\right).
\end{equation}
\STATE Set $t=t-1$.
\ENDWHILE
\STATE Choose $\phi_0\in \mathbb{R}$ such that $e^{\I \phi_0}=P(1)$. Set $\phi_j=\frac{\pi}{2}$ for $j=1,3,\dots,t-1$, $\phi_{j'}=-\frac{\pi}{2}$ for $j'=2,4,\dots,t$.
\STATE \textbf{Output:} QSP phase factors $\Phi=(\phi_0,\dots,\phi_\qspdeg)$ satisfying 
\begin{equation}
        U_\Phi(x) = e^{\I \phi_0 \sigma_z} \prod_{j=1}^{\qspdeg} \left[ W(x) e^{\I \phi_j \sigma_z} \right]
        = \left( \begin{array}{cc}
        \tilde P(x)+\I B(x) & (\I\tilde Q(x)- C(x)) \sqrt{1 - x^2}\\
        (\I \tilde Q(x)+ C(x)) \sqrt{1 - x^2} &  \tilde P(x)-\I B(x)
        \end{array} \right)
\end{equation}

\end{algorithmic}
\end{algorithm}

\FloatBarrier

The Haah method proposed in \cite{Haah2019} is summarized in \cref{alg:Haah}. Here a Laurent polynomial of degree $\qspdeg$ represents polynomials having the form $P(z)=\sum_{j=-\qspdeg}^\qspdeg p_jz^j,\,p_j\in\mathbb{C}\,, |p_\qspdeg|+|p_{-\qspdeg}|\neq 0$.  A complex-valued function $P$ is said to be real-on-circle if $P(z)\in\mathbb{R},\,\forall |z|=1$.

Suppose two real polynomials $\tilde P(x)$ and $\tilde Q(x)$ satisfy the requirements of \cref{alg:GSLW}, they can be converted to desired input of \cref{alg:Haah} through the formula
\begin{equation}
A(z)=\tilde P\left(\frac{z+z^{-1}}{2}\right), \quad B(z)=\frac{z-z^{-1}}{2\I}\tilde Q\left(\frac{z+z^{-1}}{2}\right).
\end{equation}
If $A(z)$ and $B(z)$ are generated by this formula, we may only compute $\qspdeg+1$ terms $E_0,E_1(t),\dots,E_\qspdeg(t)$ from coefficients $C_{2k}^{2\qspdeg},k=-\qspdeg,-\qspdeg+2,\dots,\qspdeg$ such that
\begin{equation}
A(z)+\I B(z)\approx \bra{+}E_0E_{1}(z)\cdots E_{\qspdeg}(z)\ket{+}, \quad \forall  z\in U(1).
\end{equation}
\cite{Haah2019} proved that in this case matrix $P_j$ computed in the algorithm are of form 
\begin{equation}
P_j=e^{\I \sigma_z\hat \phi_j/2}\ket{+}\bra{+}e^{-\I \sigma_z\hat \phi_j/2},\,j=1,\dots,2\qspdeg,
\end{equation}
and there exists $\hat \phi_0$ such that $E_0=e^{\I\sigma_z \hat \phi_0}$. The transformation formula between $\hat \Phi=(\hat \phi_0,\dots,\hat \phi_{\qspdeg})$ and QSP phase factors $\Phi$ are given in \cref{app:qsp_unique}. In practice we take $B(z)=0$ since we are not interested in the second polynomial $\tilde Q$. As the rational approximation procedure in Step 1 is designed to bound the error theoretically and hard to implement,  in practice we round the coefficients of $(1-\epsilon/3)A(z)$ and $(1-\epsilon/3)B(z)$ with small magnitude to zero instead of taking rational approximation.

\begin{algorithm}[htbp] 
\caption{Haah method}
\label{alg:Haah}
\begin{algorithmic} 
\STATE \textbf{Input:} 
A real parameter $\epsilon\in(0,0.1)$, a nonnegative integer $R$ indicates the number of bits on which high-precision arithmetic is performed and a complex-valued Laurent polynomial $A(e^{\I\theta})+\I B(e^{\I\theta})=\sum_{k=-\qspdeg}^\qspdeg\zeta_k e^{\I k\theta}$ such that 

(1) $A$ and $B$ are real-on-circle polynomials,

(2) $|A(e^{\I\theta})|^2+|B(e^{\I\theta})|^2\le 1,\,\forall \theta\in \mathbb{R}$, 

(3)$A(e^{\I \theta})$ and $B(e^{\I \theta})$ have definite parity as a function of $\theta$.  
\vspace{1em}
\STATE \textbf{Step 1:}  Denote $\qspdeg=\deg(A)$. Taking rational approximations of each coefficient of $(1-\epsilon/3)A(z)$ and $(1-\epsilon/3)B(z)$ up to error $\frac{\epsilon}{30\qspdeg}$. Coefficients with magnitude smaller than $\frac{\epsilon}{30\qspdeg}$ should be rounded to zero. Parity properties of $A$ and $B$ should be kept during rounding procedure. Denote resulting rational real-on-circle polynomials as $a(z)$ and $b(z)$, respectively. Coefficients of $a$ and $b$ should be store as rational numbers. Denote $n=\deg(a)$ and $n'=\deg(1-a(z)^2-b(z)^2)$.

\STATE \textbf{Step 2:} Find all roots of $1-a(z)^2-b(z)^2$. Denote $S$ as the multiset that contains roots of $1-a(z)^2-b(z)^2$ with their algebraic multiplicity.

\STATE \textbf{Step 3:} Define $e(z)=z^{-\lfloor \frac{n'}{2}\rfloor}\prod_{\substack{s\in S\\|s|<1 }}(z-s)$ and constant $\alpha=\frac{1-a(z)^2-b(z)^2}{e(z)e(1/z)}$. Define complementary polynomials $c(z)$ and $d(z)$ as 
\begin{equation}
c(z)=\sqrt \alpha\frac{e(z)-e(1/z)}{2\I},\,\,\,\,d(z)=\sqrt \alpha\frac{e(z)+e(1/z)}{2}.
\end{equation}
Evaluate $c(z)$ and $d(z)$ on $D=2^{\lceil \log_2(2n+1) \rceil}$ points
\begin{equation}
\{e^{2\pi \I k/D}|k=0,\dots,D-1\}
\end{equation}
by computing $e(z)$ and $e(1/z)$ via factorized form rather than direct expansion.

\STATE \textbf{Step 4:} Compute $2$-by-$2$ complex matrices $C_{2k}^{2n},-n\le k\le n$ such that $$\sum_{k=-n}^n C^{2n}_{2k}z^k=a(z)I+b(z)\I\sigma_x+c(z)\I\sigma_y+d(z)\I\sigma_z$$ via discrete fast Fourier transform.

\STATE \textbf{Step 5:}
\FOR{$m=2n,2n-1,\dots,1$}

Compute 
\begin{equation}
P_m=\frac{(C_m^m)^\dagger C_m^m}{\text{Tr}((C_m^m)^\dagger C_m^m)},\,Q_m=\frac{(C_{-m}^m)^\dagger C_{-m}^m}
{\text{Tr}((C_{-m}^m)^\dagger C_{-m}^m)}.
\end{equation}

Define $E_m(z)=zP_m+z^{-1}(I-P_m)$. Compute coefficients 
\begin{equation}
C_k^{m-1}= C_{k-1}^mQ_m+C_{k+1}^mP_m,\,k=-m+1,-m+3,\dots,m-3,m-1.
\end{equation}
\ENDFOR
\STATE Define $E_0=C_0^0$.

\STATE \textbf{Output:} $E_0,E_1(z),\dots,E_{2n}(z)$ satisfying
\begin{equation}
|A(z^2)+\I B(z^2)-\bra{+}E_0E_{1}(z)\cdots E_{2n}(z)\ket{+}|\le \epsilon, \quad \forall  z\in U(1).
\end{equation}
\end{algorithmic}
\end{algorithm}

\FloatBarrier

\section{Proof of \cref{thm:qspcheby}}\label{sec:proofdecay}
    First consider $\qspdeg^\prime = 2\qspdeg+2$. According to \cref{lma:rot}, it is equivalent to prove
    \begin{equation}
        \norm{\Im\left[\langle 0 | U_{\wt{\Phi}}(x) | 0 \rangle \right] + \prod_{j = 0}^{2\qspdeg+1} \cos\left(\wt{\phi}_j\right) \cdot \sum_{j = 0}^\qspdeg \left(-2\tan\left(\wt{\phi}_j\right)\right) T_{2\qspdeg+1-2j}(x)}_\infty \leq \frac{1}{6} \norm{\wt{\Phi}}_1^3 + \Or\left(\norm{\wt{\Phi}}_1^5\right)
    \end{equation}
    For simplicity, we drop the tilde in phase factors. Divide the QSP phase factors into two groups symmetrically, $\Phi_l = (\phi_0, \cdots, \phi_\qspdeg),\ \Phi_r = \Phi_l^-$, Then, $U_\Phi(x)$ can be expressed in terms of the product of two QSP matrices,
        \begin{equation}
            \label{eq:app:qspcheby:qspexp}
            U_\Phi(x) = U_{\Phi_l}(x) W(x) U_{\Phi_r}(x)=e^{\I \phi_0\sigma_z} \prod_{j = 1}^\qspdeg \left[W(x) e^{\I \phi_j \sigma_z}\right] W(x) \prod_{j = 0}^{\qspdeg-1} \left[ e^{\I \phi_{\qspdeg-j} \sigma_z} W(x) \right] e^{\I \phi_0 \sigma_z}.
        \end{equation}
        Each QSP unitary can be equivalently written as
        \begin{equation}
                U_{\Phi_l}(x) = \left[\prod_{j = 0}^\qspdeg \cos(\phi_j)\right] \left( 1 + \I t_0\sigma_z \right) \prod_{j=1}^\qspdeg \left[ W(x) \left( 1 + \I t_j \sigma_z \right) \right],
        \end{equation}
        where $t_j := \tan(\phi_j) \sim \mathcal{O}(\phi_j)$. Then, the contributions up to $\mathcal{O}(\norm{\Phi}_1^4)$ come from selecting up to three $\sigma_z$'s in the expansion,           
\begin{equation}
        \label{eq:app:qspcheby:qspexptruncate}
        \begin{split}
            \frac{U_{\Phi_l}(x)}{\prod_{j = 0}^\qspdeg \cos(\phi_j)} =& W(x)^\qspdeg + \sum_{j = 0}^\qspdeg \I t_j \sigma_z W(x)^{\qspdeg-2j} - \sum_{j_1 < j_2} t_{j_1} t_{j_2} W(x)^{\qspdeg - 2(j_2 - j_1)}\\
            &- \sum_{j_1 < j_2 < j_3} \I t_{j_1} t_{j_2} t_{j_3} \sigma_z W(x)^{\qspdeg - 2(j_1 + j_3 - j_2)} + \mathcal{O}(\norm{\Phi}_1^4).
        \end{split}
        \end{equation}
Here we have used the following relation repeatedly
\[
W(x)\sigma_z=\sigma_z W(x)^{-1}.
\]

After taking imaginary part of the upper-left component in \cref{eq:app:qspcheby:qspexp}, it is evident that only odd orders in $\phi_j$'s have nonvanishing contributions according to \cref{eq:app:qspcheby:qspexptruncate}.
Furthermore, using that $U_{\Phi_r}=U_{\Phi_l}^{\top}$, we have
\begin{equation}
        \label{eq:expand}
        \begin{split}
            \frac{\Im\left[\left\langle 0 \left|U_\Phi(x) \right| 0 \right\rangle \right]}{\prod_{j = 0}^{2\qspdeg+1} \cos(\phi_j)} =& \sum_{j = 0}^\qspdeg 2 t_j T_{2\qspdeg+1-2j}(x) - \sum_{j = 0}^\qspdeg \sum_{j_1 < j_2} 2 t_j t_{j_1} t_{j_2} T_{2\qspdeg+1 - 2(j+j_2-j_1)}(x)\\
            &- \sum_{j_1 < j_2 < j_3} 2 t_{j_1} t_{j_2} t_{j_3} T_{2\qspdeg + 1 - 2(j_1 + j_3 - j_2)}(x) + \mathcal{O}(\norm{\Phi}_1^5).
        \end{split}
        \end{equation}
        Let $s_j := \sin(\phi_j)$. It implies the expected bound, 
        \begin{equation}
            \begin{split}
                \norm{\Im\left[\langle 0 | U_\Phi(x) | 0 \rangle \right] - \prod_{j = 0}^{2\qspdeg+1} \cos\left(\phi_j\right) \cdot \sum_{j = 0}^\qspdeg 2\tan\left(\phi_j\right) T_{2\qspdeg+1-2j}(x)}_\infty &\leq 2 \left( \frac{1}{2} + \frac{1}{6} \right) \sum_{j_1,j_2,j_3 = 0}^\qspdeg \abs{s_{j_1} s_{j_2} s_{j_3}} + \mathcal{O}(\norm{\Phi}_1^5)\\
                & \le \frac{4}{3} \left( \norm{\Phi}_1 / 2\right)^3 + \Or\left(\norm{\Phi}_1^5\right) = \frac{1}{6} \norm{\Phi}_1^3 + \Or\left(\norm{\Phi}_1^5\right).
            \end{split}
        \end{equation}
    This proves \cref{eqn:decay_bound} for even $\qspdeg^\prime$.
       
    Then prove the case $\qspdeg^\prime = 2\qspdeg+1$. The QSP unitary is again divided symmetrically and we drop the tilde in phase factors for simplicity. Define $\Phi_l = (\phi_0, \cdots, \phi_{\qspdeg-1}),\ \Phi_r = \Phi_l^-$
    \begin{equation}
        \label{eq:app:qspcheby:qspexp:odd}
        U_\Phi(x) = U_{\Phi_l}(x) W(x) e^{\I \phi_\qspdeg \sigma_z} W(x) U_{\Phi_r}(x) = \underbrace{\cos(\phi_\qspdeg) U_{\Phi_l}(x) W(x)^2 U_{\Phi_r}(x)}_{\textcircled{1}} + \underbrace{\I \sin(\phi_\qspdeg) U_{\Phi_l}(x) W(x) \sigma_z W(x) U_{\Phi_r}(x)}_{\textcircled{2}}
    \end{equation}
    Similar to expansion in \cref{eq:app:qspcheby:qspexptruncate}, we conclude the following bounds,
\begin{equation}
        \begin{split}
            & \norm{\Im\left[\langle0|\textcircled{1}|0\rangle\right] - \prod_{j=0}^{2\qspdeg}\cos(\phi_j) \cdot \sum_{j=0}^{\qspdeg-1} 2t_j T_{2\qspdeg-2j}(x)}_\infty \leq 2 \left(\frac{1}{2} + \frac{1}{6}\right)\left(\sum_{j=0}^{\qspdeg-1} \abs{s_j} \right)^3 + \Or\left(\norm{\Phi}_1^5\right) \le \frac{1}{6} \norm{\Phi}_1^3 + \Or\left(\norm{\Phi}_1^5\right),\\
            & \norm{\Im\left[\langle0|\textcircled{2}|0\rangle\right] - \prod_{j=0}^{2\qspdeg}\cos(\phi_j) \cdot \tan(\phi_\qspdeg)}_\infty \leq \frac{s_\qspdeg}{4} \norm{\Phi}_1 ^2 + \frac{s_\qspdeg}{48}\norm{\Phi}_1^4 + \Or(\norm{\Phi}_1^6)  \le \frac{1}{4} \norm{\Phi}_1^3 + \Or\left(\norm{\Phi}_1^5\right).
        \end{split}
    \end{equation}
    Using the triangle inequality, we prove \cref{eqn:decay_bound} when $\qspdeg^\prime$ is odd.

\end{document}